\documentclass[conference,10pt]{IEEEtran}
\IEEEoverridecommandlockouts
\usepackage{cite}
\usepackage{amsmath,amssymb,amsfonts,amsthm}
\usepackage{algorithmic}
\usepackage{graphicx}
\usepackage{mathrsfs}
\usepackage{textcomp}
\usepackage{xcolor}
\usepackage{romannum}
\usepackage{multirow}
\usepackage{algorithm,algorithmic}
\usepackage{subfigure}
\usepackage{cleveref}
\usepackage{outlines}
\usepackage{caption,setspace}
\usepackage{adjustbox}
\usepackage{amsmath}
\usepackage{bm}
\usepackage{comment}
\usepackage{bbm}
\theoremstyle{definition} 
\newtheorem{definition}{Definition}
\newtheorem{remark}{Remark}
\newtheorem{theorem}{Theorem}
\newtheorem{lemma}{Lemma}

\newtheorem{proposition}{Proposition}
\newtheorem{corollary}{Corollary}
\usepackage[font=footnotesize]{caption}
\usepackage{graphicx}
\usepackage{epstopdf}
\usepackage{graphicx}

\def\BibTeX{{\rm B\kern-.05em{\sc i\kern-.025em b}\kern-.08em
    T\kern-.1667em\lower.7ex\hbox{E}\kern-.125emX}}
\begin{document}

\title{Channel-Adaptive Edge AI: Maximizing Inference Throughput by Adapting Computational Complexity to Channel States \\

\thanks{
J. Zhang, J. Huang, and K. Huang are with the Department of Electrical and Computer Engineering, The University of Hong Kong, Hong Kong. Emails: jrzhang@eee.hku.hk, jianhaoh@hku.hk and huangkb@eee.hku.hk (Corresponding authors: J. Huang and K. Huang).}
}

\author{\IEEEauthorblockN{Jierui Zhang, \textit{Graduate Student Member, IEEE}, Jianhao Huang, \textit{Member, IEEE}, and Kaibin Huang, \textit{Fellow, IEEE}}}

\maketitle

\pagenumbering{arabic}

\begin{abstract}
\emph{Integrated communication and computation} (IC$^2$) has emerged as a new paradigm for enabling efficient edge inference in sixth-generation (6G) networks. However, the design of IC$^2$ technologies is hindered by the lack of a tractable theoretical framework for characterizing \emph{end-to-end} (E2E) inference performance. The metric is highly complicated as it needs to account for both channel distortion and artificial intelligence (AI) model architecture and computational complexity. In this work, we address this challenge by developing a tractable analytical model for E2E inference accuracy and leveraging it to design a \emph{channel-adaptive AI} algorithm that maximizes inference throughput, referred to as the edge processing rate (EPR), under latency and accuracy constraints. Specifically, we consider an edge inference system in which a server deploys a backbone model with early exit, which enables flexible computational complexity, to perform inference on data features transmitted by a mobile device. The proposed accuracy model characterizes high-dimensional feature distributions in the angular domain using a Mixture of von Mises (MvM) distribution. This leads to a desired closed-form expression for inference accuracy as a function of quantization bit-width and model traversal depth, which represent channel distortion and computational complexity, respectively. Building upon this accuracy model, we formulate and solve the EPR maximization problem under joint latency and accuracy constraints, leading to a channel-adaptive AI algorithm that achieves full IC$^2$ integration. The proposed algorithm jointly adapts transmit-side feature compression and receive-side model complexity according to channel conditions to maximize overall efficiency and inference throughput. Experimental results demonstrate its superior performance as compared with fixed-complexity counterparts in terms of EPR, highlighting the importance of jointly optimizing communication and computation resources.
\end{abstract}


\section{Introduction}

The \emph{sixth-generation} (6G) mobile networks are expected to enable ubiquitous edge \emph{artificial intelligence} (AI) services (also known as edge inference), supporting applications such as autonomous driving, smart cities, robotics, and healthcare\cite{shi2020communication,liu2025integrated,mao2024green,zhou2019edge}. A typical edge inference system comprises a device that extracts features from sensed data and transmits them to an edge server; the server then performs inference using a pre-trained AI model\cite{shao2020communication,wen2023SCC}. To meet the stringent performance demands of 6G, researchers have adopted an \emph{integrated communication-and-computation} (IC$^2$) design approach aimed at optimizing \emph{end-to-end} (E2E) metrics such as latency, accuracy, and throughput\cite{zhu2020toward}. Since both communication and computation operations significantly influence E2E performance, they are inherently coupled and require joint control. The joint control is especially critical in high-mobility scenarios, where both C$^2$ operations must be dynamically adapted to rapidly changing channel conditions to support high-throughput, low-latency edge AI services. However, due to the lack of tractable design methods, existing techniques predominantly enable only partial C$^2$ adaptation, namely either adaptive communication or adaptive computation but not both. To bridge this gap, this paper proposes a tractable IC$^2$ framework, termed \emph{channel-adaptive AI}, which jointly adjusts computational complexity and transmission compression based on \emph{channel state information} (CSI), with the goal of maximizing E2E inference throughput while satisfying constraints on latency and inference accuracy.

In classical communication theory, channel adaptation, namely adjusting the modulation and coding scheme based on the channel state, is a key mechanism for realizing the optimal rate-reliability tradeoff in the presence of fading \cite{goldsmith1997variable,goldsmith1998adaptive}. While traditional techniques are application agnostic, recent 6G research has introduced channel-adaptive transmission schemes tailored specifically to edge inference systems \cite{shi2020communication}. Such schemes aim to ensure reliable and efficient delivery of data features through mechanisms such as scheduling, bandwidth allocation, and power control \cite{yao2025energy,wang2024joint,qu2025partialloading}. To further reduce communication overhead, a device can partition a feature vector into multiple segments and transmit them progressively in order of decreasing importance until a target inference confidence is achieved \cite{lan2022progressive}. In distributed sensing-and-inference systems, the scalability of multi-access is addressed by over-the-air computation (AirComp), which exploits waveform superposition to perform in-air feature aggregation and thereby enables simultaneous access \cite{amiri2020machine,yang2020federated}. Realizing AirComp typically relies on distributed transmit-power adaptation to mitigate computation errors caused by channel distortion\cite{cao2020optimized}. Overall, these prior approaches primarily focus on communication adaptation and do not provide the full C$^2$ adaptation required to optimize E2E performance.

Recent research has begun to explore fully adaptive IC$^2$ designs. A representative class of techniques is split inference, where a global AI model is partitioned into two parts deployed at the device and the server for cooperative inference \cite{shao2020communication,yan2022optimal}. The traversal depth of the model on the device side affects the number of features to be transmitted and thus the communication overhead. It also determines the device’s computation energy consumption, while the communication energy counterpart depends on the channel condition \cite{zhou2025communication}. Combining these facts motivates the development of methods to adapt the model-splitting point to the channel state, jointly with bandwidth allocation or device selection\cite{yan2022optimal,lee2023wireless,chen2024adaptive,shao2020bottlenet++,li2024optimal}. Another main line of work focuses on deploying so-called early-exiting models in edge inference systems \cite{han2021dynamic,liu2023resource,li2019edge,wang2020dual,teerapittayanon2016branchynet,dong2022resource}. Such a model is a neural network architecture that provides intermediate outputs (``early exits'') at multiple depths. By enabling multi-stage predictions, the model offers a fine-grained mechanism for controlling accuracy, latency, and computation load, which is particularly useful for resource-constrained or real-time applications \cite{liu2023resource,bolukbasi2017adaptive}. The inherent tradeoff of an early-exiting model is that intermediate features become increasingly discriminative as the traversal depth increases, leading to higher inference accuracy but also greater computational cost. Deploying a large-scale backbone model at the server facilitates meeting heterogeneous quality-of-service and E2E latency requirements across users \cite{liu2023resource,li2019edge,zhou2025communication}. For example, in a multiuser system with heterogeneous requirements, the C$^2$ latency of different users can be jointly controlled via bandwidth allocation and exit-point selection to maximize inference throughput \cite{liu2023resource}. Owing to this flexibility, an early-exiting model is adopted in this work to realize the proposed adaptive AI framework.

\begin{figure}[t]
\centering
{\includegraphics[width=0.6\linewidth]{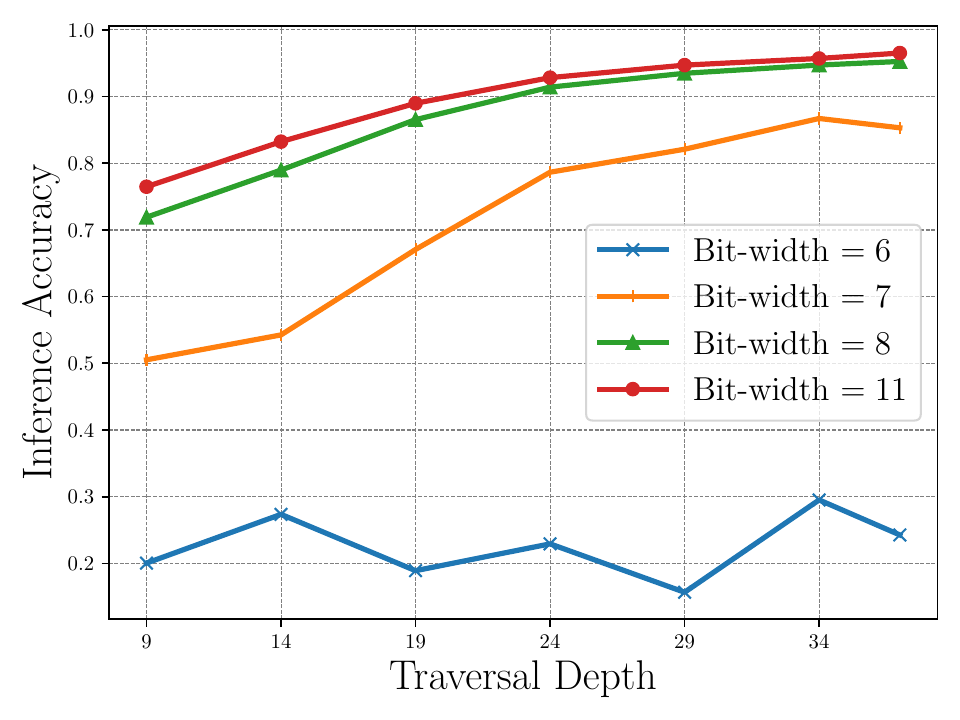}}
\caption{Dependence of inference accuracy on traversal depth and quantization bit-width. Four settings of bit-width are utilized. The model used is ResNet-152 with intermediate classifiers, and on the dataset CIFAR-10.}
\vspace{-5mm}
\captionsetup{justification=justified}
\label{motivation acc vs l}
  \end{figure}
The main challenge of designing the proposed adaptive AI techniques lies in the lack of a theoretical framework that quantifies the C$^2$ relationships among inference accuracy, computational complexity, and channel distortion. 
Such a framework would enable the derivation of a simple, closed-form control scheme to facilitate rapid channel adaptation, supporting the stringent requirements of low-latency applications. 
However, this critical framework is currently absent from the literature, as explained below. 
Prior work on edge inference typically assumes reliable communication without explicit latency guarantees (see, e.g., \cite{shao2020communication,yan2022optimal,liu2023resource}). In contrast, low-latency applications must satisfy strict latency constraints. On one hand, under unfavorable channel conditions, transmitted data must either be compressed more aggressively (incurring larger quantization errors), or suffer increased transmission errors. On the other hand, increasing the traversal depth of on-device models can improve inference robustness to such data distortions (see Fig.~\ref{motivation acc vs l}), but at the cost of higher computational complexity/latency \cite{teerapittayanon2016branchynet}. Therefore, developing efficient channel-adaptive AI techniques requires the aforementioned theoretical framework to properly account for these trade-offs. To address this issue, existing research on early-exiting models characterizes the impact of traversal depth on accuracy through empirically measured accuracy–latency pairs\cite{liu2023resource}. This enables the joint control of bandwidth allocation and exit-point selection to be formulated as a combinatorial optimization problem, which is then solved via efficient tree search algorithms. However, this method is inefficient in highly dynamic environments, as every change in channel state or user set triggers the execution of a computationally intensive tree search. Moreover, such empirical approaches are limited in more complex scenarios, such as the one considered here where multiple interdependent factors, including accuracy, complexity, and channel distortion, must be simultaneously addressed.

To address the preceding challenge, the main contribution of this work lies in developing a novel, tractable model of inference accuracy for early-exiting models and applying it to derive closed-form schemes for jointly adapting the traversal depth and quantization bit-width to the channel state. Thereby, the resulting adaptive-AI framework achieves maximum E2E throughput under constraints on air latency and inference accuracy. Focusing on a popular classification task, we consider a single-user edge inference system described shortly. In this system, a shallow model running on the edge device performs feature extraction, while a large-scale early-exiting model is deployed at the server to provide a flexible trade-off between discernibility and computational complexity/latency. Prior to transmission, the extracted features are quantized at the device with an adjustable bit-width to control communication overhead. To evaluate the E2E throughput of the system, we define the \emph{edge processing rate} (EPR) as the ratio of the number of transmitted and processed bits to the E2E inference latency, subject to a target inference accuracy and air latency. The key contributions and findings of this paper are summarized as follows.
\begin{itemize}
    \item \textbf{Tractable Modeling of Inference Accuracy}: The proposed statistical model establishes a closed-form relationship between E2E inference accuracy and the distribution of the projection of high-dimensional features onto a single-dimensional angular domain. Motivated by experimental observations, a key feature of the model is representing the feature distribution in the angular domain as a \emph{Mixture of von Mises} (MvM) distribution, with an adjustable linear-rate parameter fitted using real data. The model’s tractability is reflected in enabling the use of the maximum a posteriori (MAP) decision rule to derive closed-form expressions for inference accuracy at different traversal depths and under various levels of channel distortion. Validation against experimental results shows that the theoretical predictions closely match observed outcomes. Both theoretical and empirical results confirm the trend that inference accuracy improves with increased traversal depth at the server but degrades as quantization errors grow due to deteriorating channel conditions. This accuracy model provides the essential analytical foundation for the adaptive AI strategies developed in the subsequent work.
    \item \textbf{Channel-adaptive AI Schemes}: Building on the preceding accuracy model, we formulate an optimization problem aimed at maximizing the EPR under constraints on air latency and E2E inference accuracy. By solving this problem, a closed-form adaptive AI scheme is derived that jointly selects the quantization bit-width on the device and the traversal depth at the server based on the instantaneous channel \emph{signal-to-noise ratio} (SNR). The solution approach involves a continuous relaxation (CR) of the discrete variables, namely, traversal depth and quantization bit-width, and subsequent rounding to obtain a practical implementation. As shown by the derived scheme, it is optimal to control the bit-width linearly proportionally to the channel capacity, and select the minimal traversal depth that satisfies the accuracy constraint given the corresponding quantization distortion.
    \item \textbf{Experiments}:
    Using the CIFAR-10 dataset, we compare the proposed adaptive AI scheme with a non-adaptive baseline that employs fixed quantization bit-widths and traversal depths. The results demonstrate consistent performance improvements of the adaptive scheme across all SNR levels. Notably, at a transmit SNR of 25 dB, our scheme achieves twice the EPR of the non-adaptive approach at a target accuracy of 95\%. Additionally, by relaxing the target accuracy from 95\% to 85\%, the EPR of the proposed scheme can be further increased by 14.2\%. This highlights its flexibility in effectively balancing inference accuracy and communication efficiency.
\end{itemize}

The remainder of this paper is organized as follows. Section \ref{system model} introduces the models and metrics.
The tractable model for inference and the channel-adaptive AI are presented in Section \ref{Theoretical Model} and Section \ref{CA2I}, respectively. Experimental results are provided in Section \ref{experiments}, followed by concluding remarks in Section \ref{concluding remarks}.

\section{Models and Metrics}\label{system model}

Consider an edge inference system as shown in Fig.~\ref{system}, in which the edge device extracts the features of source data and transmits them to the edge server for inference. The inference result will subsequently be sent back to the edge device. The inference model, communication model, and performance metrics are introduced in the following subsections.

\begin{figure*}[t]
\centering
{\includegraphics[width=0.75\linewidth]{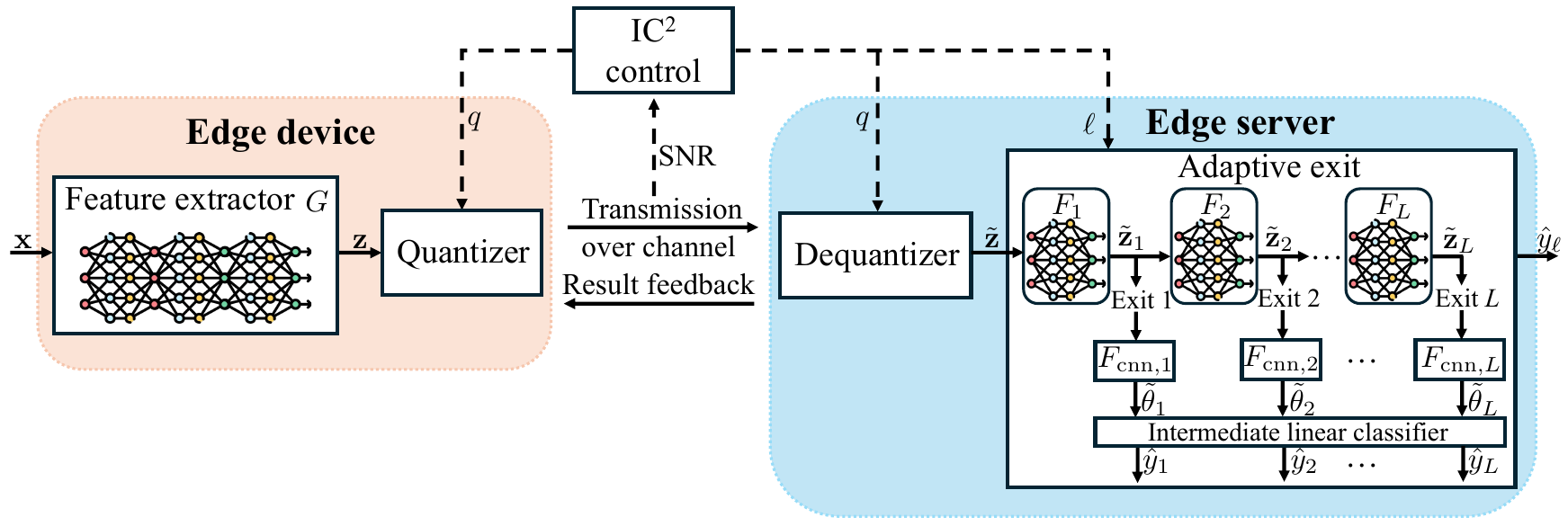}}
\caption{The system model of the channel-adaptive edge inference. The bit-width $q$ and traversal depth $\ell$ are controlled w.r.t. SNR.}
\vspace{-5mm}
\captionsetup{justification=justified}
\label{system}
  \end{figure*}

\subsection{Inference Model}\label{Inference Model}

We consider a classification problem that aims to predict the label $\hat{y}\in\mathcal{J}=\{1,2,...,J\}$ ($J\geqslant 2$) of an image sample $\mathbf{x} \in \mathbb{R}^{D}$.
A pretrained backbone model $\Phi=\{G,F\}$ is split into two parts: a shallow model $G$ is on the edge device, and a deep model $F$ with $L$ layers is on the edge server, with the $\ell$-th layer in $F$ denoted by $F_\ell$. Modern backbone models typically contain a sequence of blocks with identical architecture (e.g., a set of residual blocks within the same stage of the ResNet\cite{he2016deep}). $G$ serves as a feature extractor to transform the sample $\mathbf{x}$ to latent features of $\mathbf{z}=G(\mathbf{x}) \in \mathbb{R}^{d}$. The features are then transmitted over a channel and received at the edge server, as $\tilde{\mathbf{z}}$. At the edge server, traversing the model $F$, the received features $\tilde{\mathbf{z}}$ are gradually transformed to become more discriminative (see e.g., \cite{tishby2015deep}), thereby resulting in increasingly higher inference accuracy.

\subsubsection{Adaptive Exit}

Controllable inference latency can be realized by early exiting from the server model $F$ \cite{liu2023resource,bolukbasi2017adaptive,laskaridis2021adaptive}.
The specific $\ell$ for an inference task is determined by the channel state and target accuracy to realize the concept of \emph{adaptive exit}. After traversing the first $\ell$ layers, $\tilde{\mathbf{z}}$ become $\tilde{\mathbf{z}}_\ell=(F_\ell \circ \cdots \circ F_1)(\tilde{\mathbf{z}})$.

\subsubsection{Intermediate Classifiers}

Each early-exit point of the backbone model is attached to an intermediate classifier, denoted by $F_{\mathrm{cnn},\ell}$, for exiting at $F_\ell$. Each intermediate classifier is composed of three convolutional layers followed by fully connected layers. Profiling shows that one intermediate classifier introduces approximately \(4.6\%\) additional \emph{floating point operations} (FLOPs) and \(6.6\%\) additional parameter storage relative to the backbone. During each inference task, only the intermediate classifier at the selected exit point is executed. Note that the conventional design outputs the predicted label via a Softmax classifier. In contrast, for the proposed design, $F_{\mathrm{cnn},\ell}(\cdot)=(\mathrm{atan}2\circ \phi_\ell)(\cdot)$ first transforms $\tilde{\mathbf{z}}_\ell$ into a 2D feature vector via the nonlinear projection $\phi_\ell$, and then outputs the associated angular feature via the $\mathrm{atan}2$ function\footnote{$\mathrm{atan}2(y,x)\triangleq \mathrm{arg}(x+iy)$ is the argument of the complex number $x+iy$ \cite{ahlfors1979complex}.}: $\tilde{\theta}_\ell=F_{\mathrm{cnn},\ell}(\tilde{\mathbf{z}}_\ell)\in(-\pi,\pi]$.
The angular features from different classes exhibit distinct distributions in the angular space. Denote the conditional probability density function (PDF) of $\tilde{\theta}_\ell$ given class $j\in\mathcal{J}$ as $p_\ell(\tilde{\theta}_\ell|j)$, which is modeled in the next section. The \emph{maximum a posteriori} (MAP) criterion is employed to estimate the class (i.e., label) of $\tilde{\theta}_\ell$ as $\hat{y}_\ell=\arg\max_{j\in\mathcal{J}}\Pr(j|\tilde{\theta}_\ell)$, which, under the assumption of uniform priors, simplifies to
\begin{equation}
\hat{y}_\ell=\arg\max_{j\in\mathcal{J}}p_\ell(\tilde{\theta}_\ell|j). \label{classifier}
\end{equation}

\subsubsection{Backbone and Classifier Training}

First, the backbone model $F$ is trained using the training set $\{(\mathbf{x}_i,y_i)\}_{i=1}^{N}$ with $\mathbf{x}_i$ being the $i$-th sample, $y_i\in\mathcal{J}$ being its ground truth label, and $N$ being the sample number, and the \emph{cross entropy} as the loss function \cite{rumelhart1986learning}. 
As for the proposed intermediate classifier $F_{\mathrm{cnn},\ell}$, the training procedure is as follows. The training samples $\{\mathbf{x}_i\}_{i=1}^{N}$ are transformed into high-dimensional intermediate features $\{\tilde{\mathbf{z}}_{\ell,i}\}_{i=1}^{N}$ using $G$, followed by $F_\ell \circ \cdots \circ F_1$.
To facilitate classification and for ease of modeling, we resort to \emph{representation learning} to train $F_{\mathrm{cnn},\ell}$ for mapping $\{\tilde{\mathbf{z}}_{\ell,i}\}_{i=1}^{i=N}$ to angles $\{\tilde{\theta}_{\ell,i}\}_{i=1}^{i=N}$.
Specifically, for each class $j\in\mathcal{J}$, we assign a predefined class centroid $\mu_j\in(-\pi,\pi]$.
In the angular space, the \emph{angular distance} is measured as 
\begin{equation} d(\beta_1,\beta_2)=\min\big(|\beta_1-\beta_2|,2\pi-|\beta_1-\beta_2|\big).\label{angular distance}
\end{equation}
Then, we define the loss function over the training set $\{(\tilde{\mathbf{z}}_{\ell,i},\mu_{y_i})\}_{i=1}^{N}$ as
\begin{equation}
    \mathscr{L} \stackrel{\triangle}= 1 - \frac{1}{N} \sum_{i=1}^{N} \cos\Big(d(\tilde{\theta}_{\ell,i},\mu_{y_i})\Big).
\end{equation}
Then, minimizing this loss function by training encourages angular samples to cluster around their centroids corresponding to their ground-truth labels. We demonstrate that employing this loss does not degrade learning performance compared to standard cross entropy; rather, it enhances training robustness, thereby leading to improved overall results, as shown in Fig.~\ref{softmax vs angular}.

\begin{figure}[t]
\centering
{\includegraphics[width=0.6\linewidth]{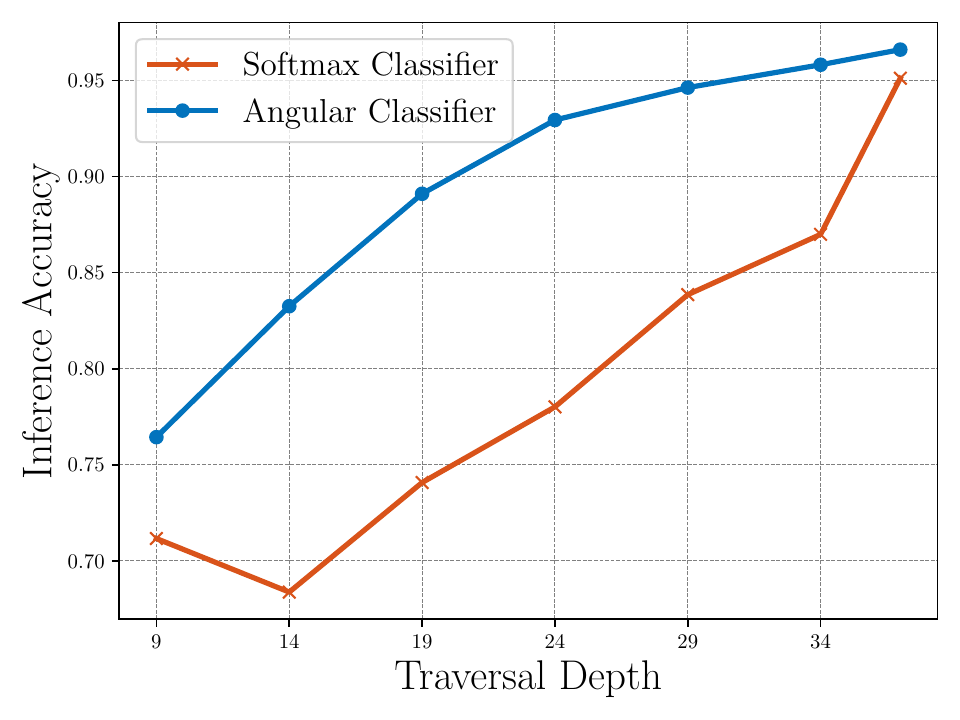}}
\caption{Comparison between the standard Softmax classifier (trained with cross entropy) and the proposed angular classifier. Both classifiers are trained for 50 epochs.}
\vspace{-5mm}
\captionsetup{justification=justified}
\label{softmax vs angular}
  \end{figure}

\subsection{Communication Model}

Consider the transmission of the feature vector, $\mathbf{z}$, by the device. Each element of $\mathbf{z}$, say $z$, is quantized into $q$ bits by uniform quantization\cite{gersho1977quantization}. Specifically, let $[c_{\mathrm{min}},c_{\mathrm{max}}]$ be the quantization range.
Let the quantization step be $\Delta=\frac{c_{\mathrm{max}}-c_{\mathrm{min}}}{2^q}$.
Then the quantized $z$ is given as
$\tilde{z}=c_{\mathrm{min}}+\Delta\cdot\big( \lfloor \frac{z-c_{\mathrm{min}}}{\Delta}\rfloor+\frac{1}{2}\big)$,
where $\lfloor\cdot\rfloor$ is the floor function.
The quantization level $\lfloor \frac{z-c_{\mathrm{min}}}{\Delta}\rfloor$ is mapped to a binary code of $q$ bits. The binary codes of all features are concatenated to form a bit stream, which is then modulated and transmitted.
The $i$-th symbol received at the server is given as
\begin{equation} y_{i}=h\sqrt{p} x_{i}+n_{i},\label{x_transmission_digital}
\end{equation}
where
$x_{i}$ is the $i$-th transmitted symbol, $p$ is the transmit power, $h$ is the channel coefficient, and $n_{i}\sim \mathcal{CN}(0,\nu_n^2)$ is the \emph{independent and identically distributed} (IID) \emph{additive white Gaussian noise} (AWGN).
It is assumed that $h$ remains constant over the duration of transmitting a single feature vector.
The communication rate is calculated using the Shannon equation as $r = B\log_2(1+\gamma)$, where $B$ is the bandwidth and $\gamma = \frac{p|h|^2}{\nu_n^2}$ is the receive SNR \cite{shannon1948mathematical}.
Then, the communication latency for transmitting $\mathbf{z}$, denoted as $T_{\mathrm{comm}}$, is
\begin{equation}
    T_{\mathrm{comm}}=\frac{dq}{B\log_2{(1+\gamma})}.\label{commu latency}
\end{equation}
To support a low-latency application, a latency constraint is enforced:
    $T_{\mathrm{comm}}\leq T_{\mathrm{max}},\label{commu constraint}$
where $T_{\mathrm{max}}$ is the maximal air latency.
The edge server demodulates the received symbols to recover the bit stream and reconstruct the feature vector as $\tilde{\mathbf{z}}$. The overhead of downloading the inference result (i.e., scalar label) back to the device is negligible.

\subsection{Performance Metrics}

The system performance can be evaluated using two metrics: inference accuracy and throughput (i.e., EPR). The inference accuracy to be modeled in the next section is denoted as $P(q,\ell)$, which is affected by the bit-width $q$ and traversal depth $\ell$ at the edge server. Next, the EPR is defined as follows.
The E2E inference latency is the summation of communication latency and computation latency:
$T_{\mathrm{total}}=T_{\mathrm{comm}}+T_{\mathrm{comp}}$,
where $T_{\mathrm{comp}}$ combines inference time at both the device and server. Mathematically,
\begin{equation}
T_{\mathrm{comp}}=\frac{\Lambda}{\nu_d}+\frac{\sum_{k=1}^{\ell}\lambda(k)}{\nu_s},\label{comp latency flops}
\end{equation}
where $\lambda(k)$ is the FLOPs in layer $k$ on the server while $\Lambda$ is the FLOPs of the device model, and $\nu_d$ and $\nu_s$ are the computation speed of the device and server, respectively.
Since the layers in the backbone model have the same structure, we can write $T_{\mathrm{comp}}$ in \eqref{comp latency flops} as an affine function w.r.t. $\ell$, as \cite{yao2025energy}
\begin{equation}
T_{\mathrm{comp}}=b_1\ell+b_2,\label{comp latency}
\end{equation}
where $b_1$ and $b_2$ are constants.

\begin{definition}
    [Edge Processing Rate (EPR)]
    Given the bit-width $q$, the traversal depth $\ell$, and the inference accuracy $P(q,\ell)$, the EPR is defined as the number of bits that the edge inference system can process per unit time\footnote{The switching overhead of adapting \((q,\ell)\) is omitted because changing \(q\) only updates the quantization/bit-packing configuration, while changing \(\ell\) only selects a pre-attached exit head. The computation of \(q\) and \(\ell\) is also lightweight, as discussed in Section~IV. Moreover, in the considered quasi-static setting, the selected \((q,\ell)\) can be reused while the channel estimate remains valid.}:
\begin{equation}
\mathrm{EPR}\stackrel{\triangle}=\frac{dq}{T_{\mathrm{comm}}+T_{\mathrm{comp}}}. \label{EPE definition}
    \end{equation}
By substituting \eqref{commu latency} and \eqref{comp latency} into \eqref{EPE definition},
\begin{equation}
       \mathrm{EPR}(q,\ell)=\frac{dq}{\frac{dq}{B\log_2{(1+\gamma})}+b_1\ell+b_2}. \label{EPE ql definition}
    \end{equation}
\end{definition}

\begin{remark}
    [Comparison with conventional metrics]\mbox{}
    \begin{itemize}
    \item \textit{EPR versus communication rate:} Compared with the classic communication rate \cite{goldsmith1997variable}, EPR is constrained by both the computation and communication latency. EPR is upper bounded by the communication rate, i.e., $B\log_2{(1+\gamma})$. When the computation speed approaches infinity, the EPR approaches the communication rate. 
     \item \textit{EPR versus conventional inference throughput:} EPR generalizes conventional inference throughput measured in inferences per second by accounting for the number of transmitted feature bits. For fixed \(d\) and \(q\), EPR is equivalent to a constant-scaled version of inferences per second; when \(q\) varies, EPR better reflects the communication-computation tradeoff.
    \item \textit{EPR versus inference accuracy:} According to Fig.~\ref{motivation acc vs l} and \eqref{EPE ql definition}, it can be observed that an increased traversal depth, $\ell$, enhances the inference accuracy but lowers the EPR. When $\ell=1$, the accuracy is at its minimum while EPR is at its maximum. The reverse is true when $\ell$ is large.
\end{itemize}
\end{remark}
Overall, the design of EPR reflects the IC\(^2\) principle by jointly characterizing the transmitted feature-bit volume and the E2E inference latency.
To balance the preceding tradeoff, it is necessary to determine the optimal bit-width and traversal depth, which is the topic of Section \ref{CA2I}.

\section{Tractable Modeling of Inference Accuracy}\label{Theoretical Model}

In this section, we establish a tractable model of inference accuracy to facilitate the design of channel-adaptive AI. First, the mixture of von Mises model is proposed to represent the distribution of angular features. Then, we use the model to characterize inference accuracy in relation to traversal depth and the bit-width.

\begin{figure}[t]
\centering
        \subfigure[2D features.]{\includegraphics[width=0.48\linewidth]{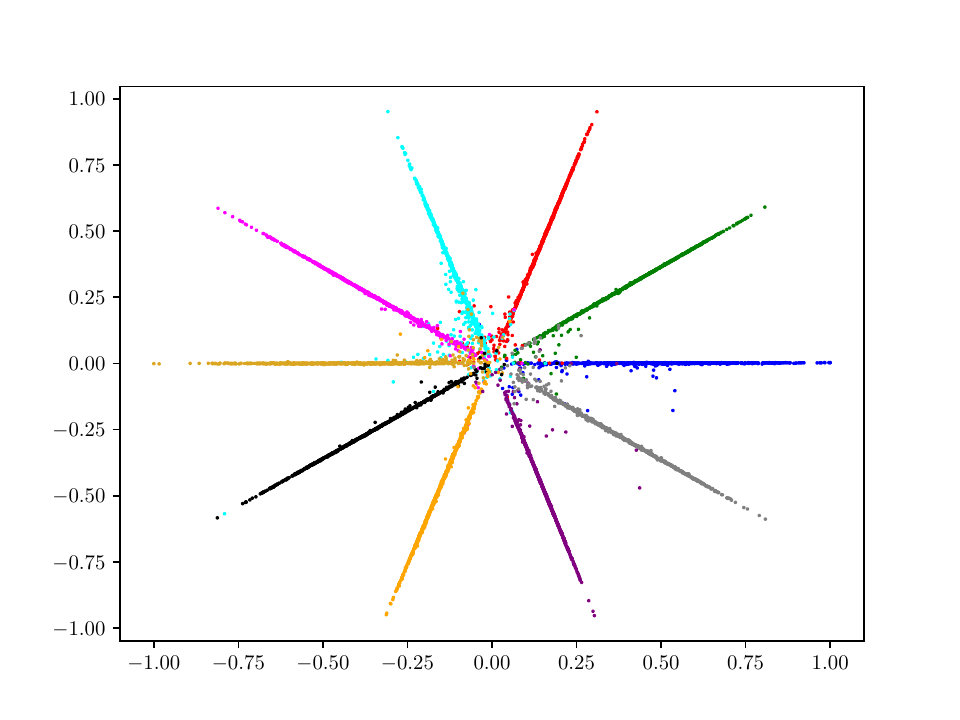}\label{2D features}}
    \subfigure[Angular features.]{\includegraphics[width=0.48\linewidth]{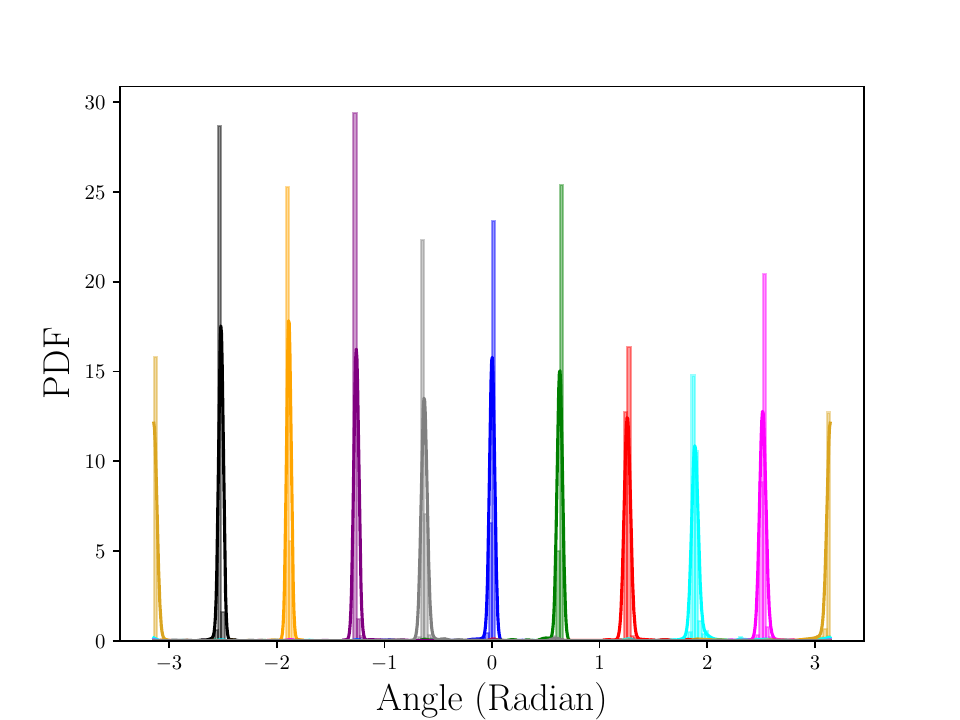}\label{1D feature dist.}}
    \caption{The distributions of the 2D and angular features on CIFAR-10. Each colored peak corresponds to one of the 10 CIFAR-10 classes, clustering around its predefined centroid. In the angular domain, the class centroids are evenly distributed and the concentration levels across different classes are similar.}
    \vspace{-5mm}
    \label{vMF dist.}
  \end{figure}

\subsection{Modeling the Distribution of Angular Features}

Consider the AI model described in Section \ref{Inference Model} and the angular features output at an arbitrary exit point. In this subsection, we describe a statistical model for their distribution, which is useful for subsequent classification accuracy analysis. First, we adopt a representation learning method from the existing literature to train the model such that the angular features form clusters in the angular domain, as illustrated in Fig.~\ref{vMF dist.} and \cite{wang2018cosface,wen2016discriminative}.\footnote{The trained backbone model with the proposed intermediate classifiers is utilized to transform the samples in the validation set. The effects of quantization are omitted. The traversal depth is set to $\ell=37$, but the distribution pattern is also valid for other values. The CIFAR-10 dataset and ResNet-152 are utilized in the demonstrations in this section, with the details described in the experimental settings in Sec.~\ref{Experimental Setting}.} One modification we make on the method is to introduce a set of given centroids into the learning process to improve classification accuracy. For simplicity, we assume uniform class priors and choose the centroids to be equally spaced in angle; for other class-prior distributions, the centroid locations can be adjusted accordingly.
Based on the above learning method, it has been shown in the literature that the angular features of a single class $j$ are approximately distributed in the angular domain (see Fig.~\ref{1D feature dist.}), according to a von Mises (vM) distribution defined by the following PDF \cite{scott2021mises,hasnat2017mises}
\begin{equation}
p_\ell(\tilde{\theta}_\ell|j)=\frac{\exp{\big(\kappa^{(j)}_{\Delta,\ell} \cos{(\tilde{\theta}_\ell-\mu
    _j)}\big)}}{2\pi I_0(\kappa^{(j)}_{\Delta,\ell})}, \ \tilde{\theta}_\ell\in(-\pi,\pi],\label{VM distribution}
\end{equation}
where $I_0(s)=\frac{1}{\pi}\int_{0}^{{\pi}}\exp{(s \cos{(x)})}\,dx$ is the modified Bessel function of the first kind and order $0$, $\mu_j$ is the centroid of the $j$-th class, and $\kappa^{(j)}_{\Delta,\ell}$ is a concentration parameter related to quantization noise and traversal depth. Note that $\kappa^{(j)}_{\Delta,\ell}$ is inversely related to dispersion, so $1/\kappa^{(j)}_{\Delta,\ell}$ plays a role analogous to variance. We denote the distribution as $\tilde{\theta}_\ell|j\sim v\mathcal{M}(\mu_j,\kappa^{(j)}_{\Delta,\ell})$. Accordingly, the multi-class angular features follow a \emph{mixture of von Mises} (MvM) distribution \cite{zhe2019directional}:
\begin{equation}
p_\ell(\tilde{\theta}_\ell)=\frac{1}{J}\sum_{j=1}^J p_\ell(\tilde{\theta}_\ell|j).\label{MvM}
\end{equation}
Most classes in practical datasets typically exhibit a similar level of concentration, as demonstrated in Fig.~\ref{1D feature dist.}. Consequently, a common simplifying assumption of uniform $\kappa_{\Delta,\ell}$, i.e., $\kappa^{(j)}_{\Delta,\ell}=\kappa_{\Delta,\ell}, \forall j$, is adopted in the literature and also made in the model in \eqref{MvM} \cite{zeng2024knowledge,hastie2009elements}. The parameter $\kappa_{\Delta,\ell}$ in (11) is the most important parameter of the MvM model. It can be estimated from a dataset using the procedure described in Appendix A. 

\subsection{Model Parameter w.r.t. Traversal Depth and Channel Distortion}

Consider the key parameter, $\kappa_{\Delta,\ell}$, of the MvM model in \eqref{MvM}.
As demonstrated in Fig.~\ref{kappa w.r.t. l and q}, it is related to traversal depth and quantization bit-width, which determine the levels of computational complexity and channel distortion, respectively.
In this subsection, we investigate these relationships. 

\begin{figure}[h]
\centering
    \subfigure[$\kappa_{\Delta,\ell}$ w.r.t. traversal depth.]    {\includegraphics[width=0.48\linewidth]{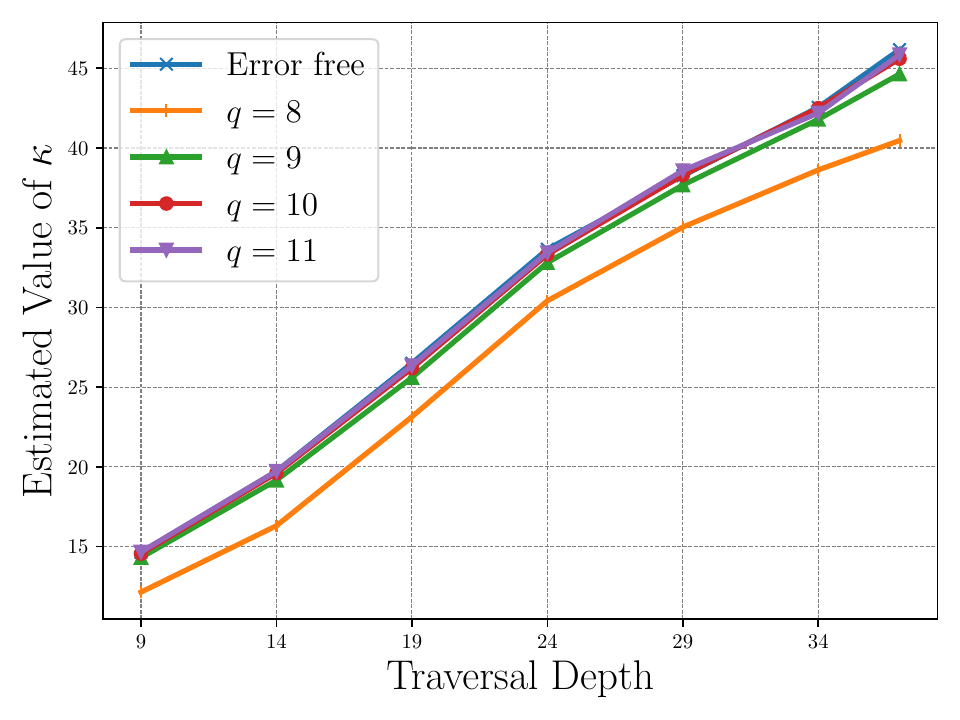}\label{kappa vs l in Q}}
    \subfigure[$\kappa_{\Delta,\ell}$ w.r.t. bit-width.]{\includegraphics[width=0.48\linewidth]{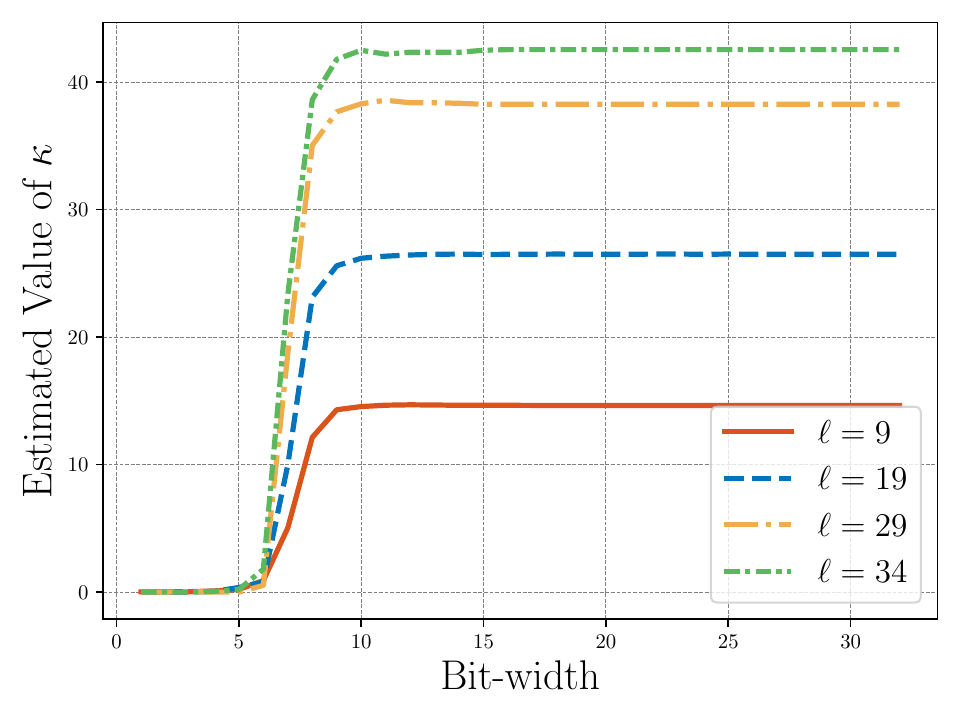}\label{kappa vs Q in l}}
\caption{$\kappa_{\Delta,\ell}$ w.r.t. traversal depth and bit-width.}\label{kappa w.r.t. l and q}
  \end{figure}

\subsubsection{Modeling Distortion w.r.t. Bit-width}

In the considered edge inference system, the distortion on the feature vector mainly comes from quantization errors, which is significant under the latency constraint. We model this distortion element-wise using a uniform distribution: $\epsilon \sim \mathcal{U}\left(-\Delta/2,\Delta/2\right)$ \cite{gray1998quantization}.\footnote{The uniform quantization-error model is adopted for simplifying exposition and analytical tractability. For non-uniform quantization, the resulting bit-width--distortion relationship in \eqref{error variance} can be replaced by an empirically measured one, while the subsequent channel-adaptive optimization framework remains applicable.}
The variance is thus given as
\begin{equation}
\sigma_\Delta^2(q)=\frac{(c_{\mathrm{max}}-c_{\mathrm{min}})^2}{12\cdot 2^{2q}}.\label{error variance}
\end{equation}
How the distortion propagates through the neural network is quantified in the sequel.

\subsubsection{$\kappa_{\Delta,\ell}$ w.r.t. Traversal Depth}

First, assume quantization errors are negligible, i.e., $\sigma_\Delta^2=0$. Let $\bar{\kappa}_\ell$ denote the concentration parameter in this scenario. Based on the empirical results shown in Fig.~\ref{kappa vs l in Q}, the relationship between $\bar{\kappa}_\ell$ and the traversal depth, $\ell$, can be modeled as the following linear function
\begin{equation}
    \bar \kappa_{\ell}=c_1 \ell + c_2,\label{k_l}
\end{equation}
where the coefficients $c_1$ and $c_2$ can be estimated by the \emph{least square method} \cite{montgomery2021introduction}.

\subsubsection{$\kappa_{\Delta,\ell}$ w.r.t. Traversal Depth and Channel Distortion}

Next, based on the preceding models, the desired relations are quantified in the following proposition.

\begin{proposition}
    \label{prop distorted distribution}
    Assume a data sample to transmit belongs to $j$, given the traversal depth, $\ell$, at the server, and quantization distortion with variance $\sigma^2_\Delta$, the channel-distorted angular feature, $\tilde{\theta}_\ell|j$, follows a vM distribution with centroid $\mu_j$ and concentration parameter $\kappa_{\Delta,\ell}$, i.e., $v\mathcal{M}(\mu_j,\kappa_{\Delta,\ell})$. $\mu_j$ is fixed from the training process and remains unchanged, and $\kappa_{\Delta,\ell}$ is given as
\begin{equation}
\kappa_{\Delta,\ell}=A^{-1}\big(A(\bar\kappa_\ell)A(\rho_{\Delta,\ell})\big),\label{distorted kappa calculation}
\end{equation}
where $A(\cdot)\stackrel{\triangle}=\frac{I_1(\cdot)}{I_0(\cdot)}$, $A^{-1}(\cdot)$ is its inversion, $I_1(s)=\frac{1}{\pi}\int_{0}^{\pi}\cos{x}\exp{(s \cos{x})}\, dx$ is the modified Bessel function of the first kind and order $1$, and  $\rho_{\Delta,\ell} = A^{-1}\Bigl(\exp\bigl(-\sigma_\Delta^2 a_\ell/2\bigr)\Bigr)$. Moreover, $a_\ell \triangleq \mathbb{E}_{\mathbf{z}} \left[ \mathbf{a}_\ell(\mathbf{z})^\top \mathbf{a}_\ell(\mathbf{z}) \right]$ represents the expected
squared norm of the gradient 
$\mathbf{a}_\ell \triangleq \nabla f_\ell(\mathbf{z})$, 
where the model propagation function is
$f_\ell(\cdot) \triangleq (F_{\mathrm{cnn},\ell} \circ F_\ell \circ \cdots \circ F_1)(\cdot)$. According to experimental results (see Fig.~\ref{al vs l}), $a_\ell$ decreases exponentially w.r.t. $\ell$, therefore we approximate $a_\ell$ by
\begin{equation}
    a_\ell=c_3\exp(-c_4\ell),\label{al}
\end{equation}
where $c_3$ and $c_4$ are estimated positive constants. 
\end{proposition}

\begin{proof}
    See Appendix~\ref{Proof of P1}.
\end{proof}

\begin{remark}
    [Effects of channel distortion]
The result in \eqref{distorted kappa calculation} indicates that the presence of noise always makes the $\kappa_{\Delta,\ell}$ smaller and correspondingly reduces the discernibility of data features. The reason is as follows. First of all, the function $A(\cdot)$ is monotonically increasing, so is its reverse, $A^{-1}(\cdot)$. As a result, $\rho_{\Delta,\ell}$ defined in the proposition statement is a monotone decreasing function of the error variance $\sigma_\Delta^2$. It follows that $\kappa_{\Delta,\ell}$ decreases as $\sigma_\Delta^2$ grows, as claimed earlier, resulting in lower classification accuracy. For a sanity check, as $\sigma_\Delta^2 \to 0$, we have $\rho_{\Delta,\ell} \to \infty$ and hence $A^{-1}\big(A(\bar\kappa_\ell)A(\rho_{\Delta,\ell})\big) \to A^{-1}\big(A(\bar\kappa_\ell)\cdot1\big) \to \bar\kappa_\ell$, yielding $\kappa_{\Delta,\ell}\to\bar \kappa_{\ell}$ as in the ideal case of no quantization error.
\end{remark}

\subsection{Inference Accuracy Analysis}

Building on the preceding model of angular feature distribution, the performance of the inference system is analyzed with the result stated in the following theorem.

\begin{theorem}
Consider the MvM model of angular feature distribution in \eqref{MvM} with the parameter $\kappa_{\Delta,\ell}$ given in \eqref{distorted kappa calculation} and the MAP detection method in \eqref{classifier}. The inference accuracy for channel distortion $\sigma_\Delta^2$ and traversal depth $\ell$ is given as
\begin{equation}
    P(\sigma_\Delta^2,\ell)=
    \int_{0}^{\frac{\pi}{J}}\frac{\exp{\big(\kappa_{\Delta,\ell} \cos{(x)}\big)}}{\pi I_0(\kappa_{\Delta,\ell})}\, dx,\label{accuracy}
\end{equation}
where the model parameter $\kappa_{\Delta,\ell}$ is simplified as
\[\kappa_{\Delta,\ell}=A^{-1}\big(A(c_1 \ell + c_2)\exp(-\sigma_\Delta^2 c_3\exp(-c_4\ell)/2)\big).\]
\end{theorem}

\begin{proof}
    Since all clusters are equidistant and possess identical shapes, the decision boundaries are located at the midpoints between neighboring cluster centers. The classification accuracy for a specific class is obtained by integrating its PDF over its corresponding decision region; by symmetry, this accuracy is identical across all classes. Taking class $j$ as an example, the classification accuracy is given by
\[\mathrm{Pr}\left(d(\mu_{j},\tilde{\theta}_\ell) < \frac{\pi}{J}\right) = \int_{-\frac{\pi}{J}}^{\frac{\pi}{J}}\frac{\exp{(\kappa_{\Delta,\ell} \cos x)}}{2\pi I_0(\kappa_{\Delta,\ell})}\, dx,\]
which equals \(\int_{0}^{\frac{\pi}{J}}\frac{\exp{(\kappa_{\Delta,\ell} \cos x)}}{\pi I_0(\kappa_{\Delta,\ell})}\, dx,\) due to the symmetry of the integrand. Assume equal class priors, then this expression represents the overall system accuracy, yielding \eqref{accuracy}. Finally, substituting \eqref{k_l}, \eqref{distorted kappa calculation}, and \eqref{al} into \eqref{accuracy} yields the desired result.
\end{proof}

The behavior of \eqref{accuracy} is depicted in Fig.~\ref{theory acc}, while its properties are characterized as follows.

\begin{figure}[t]
\centering
{\includegraphics[width=0.6\linewidth]{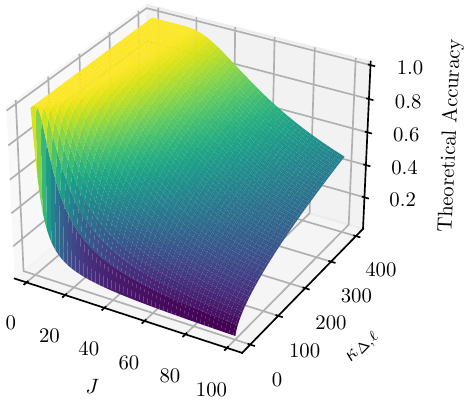}}
\caption{Inference accuracy function in \eqref{accuracy} w.r.t. the class number, $J$, and angular-feature distribution parameter, $\kappa_{\Delta,\ell}$.}
\vspace{-5mm}
\captionsetup{justification=justified}
\label{theory acc}
  \end{figure}

\subsubsection{Monotonicity}

The monotonicity of the accuracy function in \eqref{accuracy}, which is useful for the subsequent design, is stated as follows.

\begin{corollary}
    [Monotonicity of Accuracy Function]\label{monotonicity accuracy}
    The function $P(\sigma_\Delta^2,\ell)$ in \eqref{accuracy} is monotonically decreasing w.r.t. $J$ and $\sigma_\Delta^2$, and monotonically increasing w.r.t. $\kappa_{\Delta,\ell}$ and $\ell$.   
\end{corollary}

\begin{proof}
    See Appendix~\ref{proof monotonicity accuracy}.
\end{proof}

The above properties are aligned with intuition. First, as the number of object classes, $J$, increases, their discernibility drops and thus the accuracy decreases.\footnote{For large-class scenarios, the 2D angular embedding may become less effective because the angular margin decreases as the number of classes grows. A natural extension is to use a higher-dimensional hyperspherical embedding and model the features using a mixture of von Mises-Fisher distributions~\cite{mardia2009directional}. The corresponding accuracy generally requires integration over spherical decision regions, and developing such high-dimensional models is left for future work.} Second, increasing the distribution parameter, $\kappa_{\Delta,\ell}$, causes the features belonging to the same class to be more concentrated, which enhances the classification accuracy. These results suggest that one viable approach to improve the classification performance is to either increase $\kappa_{\Delta,\ell}$ or decrease $J$.
Third, Corollary~\ref{monotonicity accuracy} states that traversing more model layers at the receiver yields higher inference accuracy, while increased channel distortion (i.e., larger $\sigma_\Delta^2$) entails the opposite. This suggests higher computational complexity can help to mitigate the adverse effect of channel distortion, which is the fundamental principle underpinning the design in the next section.

\subsubsection{Asymptotics} 

To understand the limits of the inference accuracy, we derive the following results.

\begin{corollary}\label{range of accuracy}
    The range of the inference accuracy is $P\in(\frac{1}{J}, 1)$. When $\kappa_{\Delta,\ell}\to0$, $P\to\frac{1}{J}$, and when $\kappa_{\Delta,\ell}\to\infty$, $P\to1$.
\end{corollary}

\begin{proof}
    See Appendix~\ref{proof of range of accuracy}.
\end{proof}

Similar to Corollary~\ref{monotonicity accuracy}, the results in Corollary~\ref{range of accuracy} agree with intuition. Note that the accuracy of $\frac{1}{J}$ corresponds to that of random guessing, indicating zero discernibility of data features.

To better understand the asymptotic properties of inference accuracy, its scaling law is derived as shown below.

\begin{proposition}
\label{asymptotic behavior}
    In the case free of distortion ($\sigma_\Delta^2=0$), as the traversal depth $\ell \to \infty$, the accuracy follows the scaling law:
    \begin{equation}
       P(0,\ell) \sim 1-\frac{\sqrt{2}\, J}{\pi^{3/2} \sqrt{c_1\ell+c_2}} \exp\left(-\frac{\pi^2}{2 J^2} (c_1\ell+c_2) \right). \label{scaling P w.r.t. l}
    \end{equation}
\end{proposition}

\begin{proof}
    See Appendix~\ref{proof asymptotic behavior}.
\end{proof}

This result indicates that the classification error probability decays exponentially as $\ell$ grows. Therefore, increasing computational complexity is a highly effective way of enhancing inference performance.

\section{Channel-Adaptive AI Algorithm}\label{CA2I}

In the preceding section, we develop a theoretical model of inference accuracy. In this section, the model is leveraged to design the adaptive AI algorithm that maximizes throughput in terms of EPR under constraints on accuracy and latency. Finally, the algorithm performance is discussed.

\subsection{Problem Formulation}

The design of adaptive AI algorithm can be formulated as an optimization problem to maximize the EPR under constraints in inference accuracy and air latency. Mathematically,
\begin{subequations}
	\begin{align}
		\text{(P1)}~~\max_{q\in \mathcal{Q}, \ \ell\in\mathcal{L}} ~~ &  
		\mathrm{EPR}(q,\ell)\\ 
		\operatorname{ s.t. } ~~
		& P(\sigma_\Delta^2,\ell)\geq P_0,\label{acc constraint}\\
        & T_{\mathrm{comm}}\leq T_{\mathrm{max}},\label{comm constraint}
	\end{align}\label{P1}
\end{subequations}
\unskip
where $\mathrm{EPR}(q,\ell)$, $P(\sigma_\Delta^2,\ell)$, and $T_{\mathrm{comm}}$ are given by \eqref{EPE ql definition}, \eqref{accuracy}, and \eqref{commu latency}, respectively. Moreover, $P_0$ and $T_{\mathrm{max}}$ denote the target inference accuracy and maximal air
latency, respectively. Let $\mathcal{Q}\stackrel{\triangle}=\{0,1,2,\dots,Q\}$ be an available set of bit-widths, and $\mathcal{L}\stackrel{\triangle}=\{\ell_{1}, \ell_{2}, \dots,\ell_{i}\}\subseteq  \{1,2,\dots,L\}$ be the set of exit layers that supports adaptive exit with $\ell_{1}<\ell_{2}< \dots<\ell_{i}$. $\mathcal{L}$ is analogous to a modulation constellation.
The solution of (P1) is the channel-adaptive AI algorithm in the sequel, which determines the EPR-maximizing bit-width and traversal depth. 

\subsection{Channel-Adaptive AI Algorithm}

To solve (P1), we first relax the variables to be continuous and provide the corresponding solution in \ref{CQCL CA2I}. Subsequently, its relaxation is removed in \ref{CA2I restriction} to yield the desired algorithm.

\subsubsection{Continuous Relaxation}\label{CQCL CA2I}

For tractability, we first assume $q$ and $\ell$ to be continuous, and refer to the resultant algorithm as channel-adaptive AI with \emph{continuous relaxation} (CR). 
The associated EPR upper bounds its counterpart from solving (P1).
With CR, the continuous quantization can be realized as follows.
Suppose there are $d$ features to be quantized: if $\alpha d$ features are quantized into $q_0$ bits and $(1-\alpha) d$ features are quantized into ($q_0+1$) bits with $q_0$ being an integer, then the average bit-width is $q\stackrel{\triangle}=\alpha q_0+(1-\alpha)(q_0+1)=q_0+1-\alpha$.\footnote{Although $q$ is limited to rational values in this context, we still consider it to be continuous.} 
In this setting, the expected variance of quantization distortion in \eqref{error variance} becomes
\begin{equation}
\begin{split}
\sigma_\Delta^2(q)=&\alpha\frac{(c_{\mathrm{max}}-c_{\mathrm{min}})^2}{12\cdot 2^{2q_0}}+(1-\alpha)\frac{(c_{\mathrm{max}}-c_{\mathrm{min}})^2}{12\cdot 2^{2(q_0+1)}},\\
    &=\frac{1+3\alpha}{4}\frac{(c_{\mathrm{max}}-c_{\mathrm{min}})^2}{12\cdot 2^{2q_0}},\label{error variance cont.}
\end{split}
\end{equation}
where $q_0=\lfloor q \rfloor$ and $\alpha = 1-(q-q_0)$.
Note that when $\alpha=1$, \eqref{error variance cont.} reduces to \eqref{error variance}.
Under the CR, we solve (P1) as follows.
We first analyze the objective function of (P1):
\begin{equation}
    \mathrm{EPR}=\frac{dq}{T_{\mathrm{comm}}+T_{\mathrm{comp}}}
    =\frac{B\log_2{(1+\gamma})}{1+\frac{T_{\mathrm{comp}}}{T_{\mathrm{comm}}}}.\label{EPE detail}
\end{equation}
This suggests that maximizing $T_{\mathrm{comm}}$ and minimizing $T_{\mathrm{comp}}$ can maximize EPR. To achieve these two goals, we first formulate two intermediate optimization problems (P2) and (P3) as follows.

Maximizing $T_{\mathrm{comm}}$ fundamentally requires maximizing $q$. Therefore, (P2) is formulated as
\begin{subequations}
    \begin{align} 
\text{(P2)}~~\max_q& \quad q\\
\text{s.t.}& \quad \frac{dq}{B\log_2{(1+\gamma})}\leq T_{\mathrm{max}}.\label{latency constraint}
\end{align}
\end{subequations}
Next, we consider minimizing $T_{\mathrm{comp}}$, which fundamentally requires minimizing $\ell$ as $T_{\mathrm{comp}}$ increases monotonically w.r.t. $\ell$. Hence, (P3) is formulated as 
\begin{subequations}
	\begin{align}
		\text{(P3)}~~\min_{\ell} ~~ &  
		\ell\label{P3 obj}\\ 
		\operatorname{ s.t. } ~~
		& P(\sigma_\Delta^2(q^*),\ell)\geq P_0,\label{P3 cons}	\end{align}\label{P3}
\end{subequations}
\unskip
where $q^*$ is the solution to (P2).

\begin{proposition}
    [Solution to (P1)]\label{solving algorithm}
    Solving (P2) first and then solving (P3) yields the solution to (P1).
\end{proposition}

\begin{proof}
    We prove this by contradiction. Suppose we have obtained the aforementioned solution $q^*$ and $\ell^*$. If there exist some $q_0$ and $\ell_0$ such that the resulting EPR is higher, then we must either have (\romannumeral 1) $q_0=q^*$, or (\romannumeral 2) $q_0<q^*$, as a result of the constraint \eqref{comm constraint}.
In situation (\romannumeral 1), to make EPR higher, we must have $\ell_0<\ell^*$. But this would ruin the constraint \eqref{P3 cons} based on the monotonicity. Hence, we omit this situation. In situation (\romannumeral 2), the resulting $\sigma_\Delta^2$ gets higher, which makes $\ell_0>\ell^*$ according to the monotonicity of $P(\sigma_\Delta^2(q^*),\ell)$. Now that $q_0\leq q^*$ and $\ell_0>\ell^*$, they contribute to smaller communication latency and higher computation latency, and thus degrade the EPR according to \eqref{EPE detail}. A contradiction appears here. This completes the proof.
\end{proof}

The solution for (P2) can be straightforwardly obtained as
\begin{equation} 
q^* = \frac{T_{\mathrm{max}}B\log_2{(1+\gamma})} {d}.\label{Q_continuous}
\end{equation}
According to Corollary~\ref{monotonicity accuracy}, $P(\sigma_\Delta^2(q^*),\ell)$ increases monotonically w.r.t. $\ell$. Based on this monotonicity,
the solution to (P3), $\ell^*$,  can be obtained by solving equation $P(\sigma_\Delta^2(q^*),\ell)= P_0$ w.r.t. $\ell$.
This can be solved numerically via efficient binary search, and we denote the solution as $\ell^*=P^{-1}(\sigma_\Delta^2(q^*),P_0)$, where we use $P^{-1}(\cdot,\cdot)$ to denote the inverse of function $P(\sigma_\Delta^2,\ell)$ w.r.t. $\ell$. According to Proposition~\ref{solving algorithm}, the solution to (P1) is given by $\{q^*,\ell^*\}$.

\subsubsection{Practical Solution}\label{CA2I restriction}

We remove the CR of $q$ and $\ell$ to yield the practical channel-adaptive AI algorithm. Proposition~\ref{solving algorithm} poses no requirement on the variable property, and hence, it is applicable in this discrete context as well. Therefore, based on previous results, we can derive the solution to (P2) as
\begin{equation} 
q^* =  \lfloor \frac{T_{\mathrm{max}}B\log_2{(1+\gamma})} {d} \rfloor_\mathcal{Q},\label{Q_discrete}
\end{equation}
where $\lfloor x \rfloor_\mathcal{Q}$ is the floor function that selects the largest number in the set $\mathcal{Q}$ less than or equal to $x$. Since a larger bit-width \(q\) increases the number of transmitted bits and hence the communication latency, we conservatively round \(q\) downward. This guarantees that the rounded bit-width still satisfies the air-latency constraint \(T_{\mathrm{comm}}\leq T_{\max}\). Subsequently, the solution to (P3) is given by 
\begin{equation}
    \ell^*=\big\lceil P^{-1}(\sigma_\Delta^2(q^*),P_0)\big\rceil_{\mathcal{L}}\label{l_discrete} 
\end{equation}
as the solution to (P3), where $q^*$ is determined by
\eqref{Q_discrete} and $\lceil x \rceil_{\mathcal{L}}$ is the ceiling function that selects the smallest number in set $\mathcal{L}$ greater than or equal to $x$. 
Since a larger traversal depth \(\ell\) improves the inference accuracy, we round \(\ell\) upward. This ensures that the traversal depth is no smaller than the continuous solution and therefore satisfies the accuracy constraint under the proposed analytical model.
Note that when $P^{-1}(\sigma_\Delta^2(q^*),P_0)>\ell_{i}$, $\ell^*$ does not exist. In this case, we set $\ell^*=\ell_{i}$ and allow for a zero EPR. Now that we have solved (P1) under practical settings, we summarize the channel-adaptive AI algorithm in Algorithm \ref{channel-adaptive AI (evaluation phase)}.

\begin{algorithm}
 \caption{Channel-adaptive AI}
 \label{channel-adaptive AI (evaluation phase)}
 \begin{algorithmic}[1]
 \renewcommand{\algorithmicrequire}{\textbf{Input:}}
 \renewcommand{\algorithmicensure}{\textbf{Output:}}
 \REQUIRE Partitioned backbone model $\Phi=\{G,F\}$, intermediate CNNs $\{F_{\mathrm{cnn},\ell}\}_{\ell\in\mathcal{L}}$, inference tasks with $\{\mathbf{x}_i\}_{i=1}^N$.

\FOR {$i = 1$ to $N$}
    \STATE Use $G$ to transform $\mathbf{x}_i$ to $\mathbf{z}_i$.
    \STATE Calculate the $q$ and $\ell$ according to \eqref{Q_discrete} and \eqref{l_discrete}.
    \STATE Transmit $\mathbf{z}_i$ based on \eqref{x_transmission_digital}, and the received one is $\tilde{\mathbf{z}}_i$.
    \STATE Use $F_{\mathrm{cnn},\ell} \circ F_\ell \circ \cdots \circ F_1$ to transform $\tilde{\mathbf{z}}_i$ to $\tilde{\theta}_{\ell,i}$.
    \STATE Obtain the inferred label $\hat{y}_i$ based on \eqref{classifier}.
    \STATE Transmit back the $\hat{y}_i$ to edge device.
\ENDFOR
\ENSURE Inferred labels $\{\hat{y}_i\}_{i=1}^N$ for all samples.
 \end{algorithmic} 
 \end{algorithm}

\subsection{Discussion on Performance and Methodology}

\subsubsection{Throughput versus EPR}

The performance of channel-adaptive AI with CR is shown in Fig.~\ref{CQCL}. 
With a fixed accuracy threshold, EPR increases monotonically as SNR increases. Furthermore, when the accuracy threshold is lowered, EPR also increases for a given SNR. The EPR curves in Fig.~\ref{CQCL} are analogous to those of communication rate relative to SNR at different \emph{bit error rates} (BERs), where EPR corresponds to the communication rate, and $P_0$ to BER.

\subsubsection{Computational Complexity versus SNR and Accuracy}

Since the traversal depth, which determines the receiver complexity, can be calculated when the target accuracy and SNR are set, a visualization of channel-adaptive AI with CR is shown in Fig.~\ref{vis CQCL}. It can be shown that when the SNR decreases or the target accuracy increases, the required traversal depth increases, because only by trading more computation can the algorithm mitigate the distortion and achieve a satisfactory inference accuracy.

\begin{figure}[t]
\centering
        \subfigure[EPR.]{\includegraphics[width=0.47\linewidth]{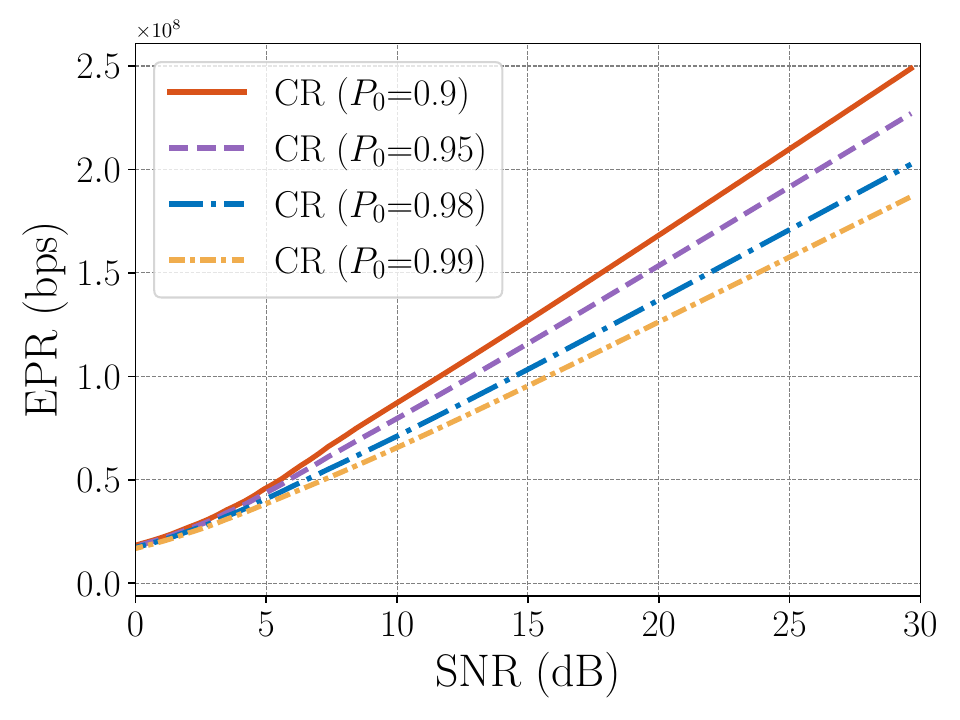}\label{CQCL}}
    \subfigure[Required traversal depth.]{\includegraphics[width=0.47\linewidth]{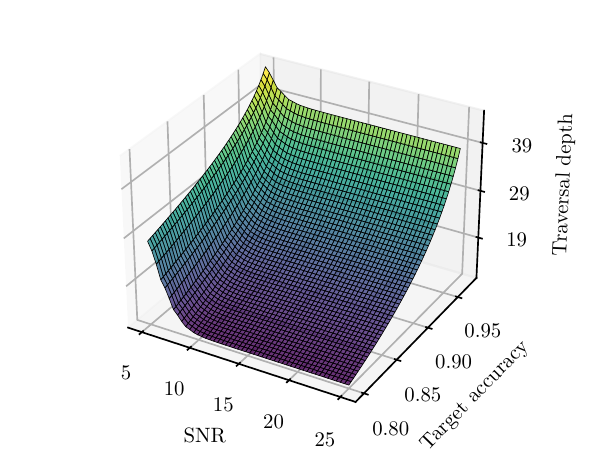}
    \label{vis CQCL}}
    \caption{Performance of channel-adaptive AI with CR.}\label{CQCL theory}
    \vspace{0mm}
  \end{figure}

\begin{table}[t]
\small
\caption{Different sets of exit layers for implementing channel-adaptive AI.}
\begin{adjustbox}{center}
\begin{tabular}{|c|c|}
\hline
\textbf{\# of exit layers} & \textbf{Set of exit layers} \\ \hline
2                & 9, 49                    \\ \hline
3                & 9, 39, 49                \\ \hline
4                & 9, 29, 39, 49            \\ \hline
5                & 9, 19, 29, 39, 49        \\ \hline
\end{tabular}
\end{adjustbox}
\label{DQDL region setting theory}
\end{table}

\subsubsection{Effect of the Set of Exit Layers}

As for the practical channel adaptive AI without CR, we set the target accuracy as $P_0=0.9$ and the considered sets of exit layers are shown in Table \ref{DQDL region setting theory}. 
As shown in Fig.~\ref{DQDL}, as the number of exit layers increases, the EPR also increases. Channel-adaptive AI with CR represents the theoretical limit for the practical one. This is due to the fact that a greater number of exit layers provides more candidates for exit, thereby increasing the potential to reduce computation latency while satisfying the accuracy requirement. Consequently, as the number of exit layers increases, the EPR of the practical channel-adaptive AI without CR can approach that with CR.
Additionally, we introduce a non-adaptive algorithm as a benchmark, where the bit-width is fixed, and only one exit layer (i.e., 49) is utilized. Since both parameters remain constant regardless of channel states, this approach is classified as non-adaptive. This is analogous to non-adaptive transmission, where the transmission rate and power are fixed \cite{goldsmith1997variable}. The curves presented in Fig.~\ref{DQDL} demonstrate the superiority of channel-adaptive AI over its non-adaptive counterpart.
A visualization of the practical channel-adaptive algorithm is shown in Fig.~\ref{vis DQDL}. 
Due to the discrete bit-width and layers, the resulting surface appears stepped, which distinguishes it from channel-adaptive AI with CR. 
Furthermore, the application of the floor function for \(q\) and the ceiling function for \(\ell\) generally results in a greater required traversal depth in practical channel-adaptive AI compared with channel-adaptive AI with CR. This is because rounding \(q\) downward reduces the bit-width, which increases the quantization distortion and therefore requires a larger traversal depth to maintain the target accuracy; the ceiling operation further ensures that the selected traversal depth is no smaller than the required continuous value. Consequently, whenever a feasible discrete solution exists, the target-accuracy constraint remains satisfied after discretization. If no feasible discrete solution exists, the current channel state is declared infeasible, and the EPR is set to zero; in this case, the worst-case accuracy shortfall is upper-bounded by \(P_0-1/J\) under the analytical model.

\begin{figure}[t]
\centering
        \subfigure[EPR.]{\includegraphics[width=0.47\linewidth]{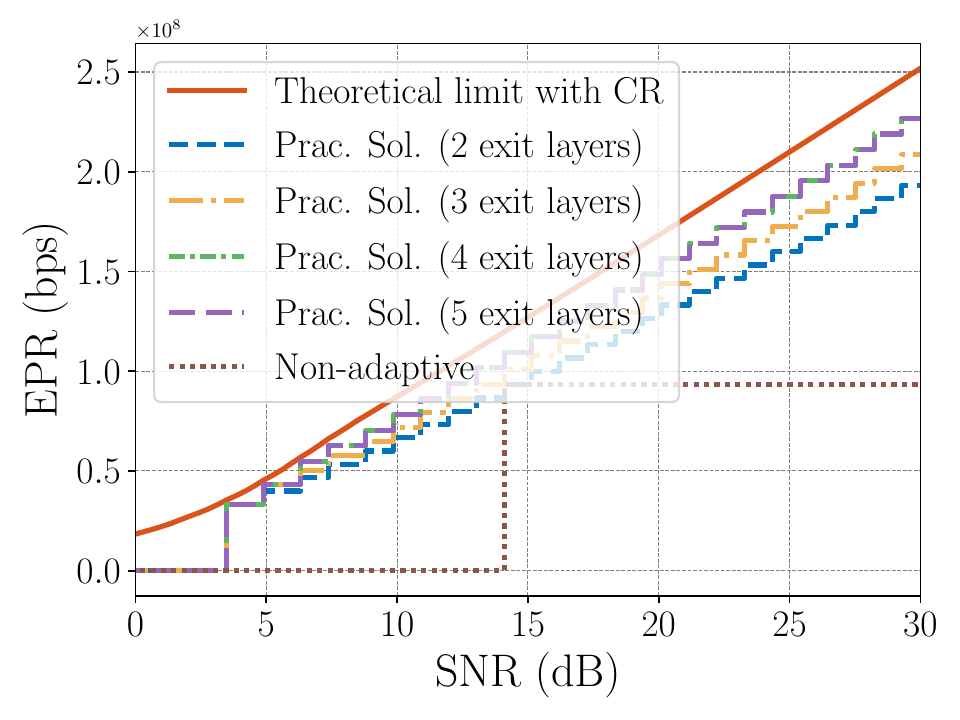}\label{DQDL}}
    \subfigure[Required traversal depth.]{\includegraphics[width=0.47\linewidth]{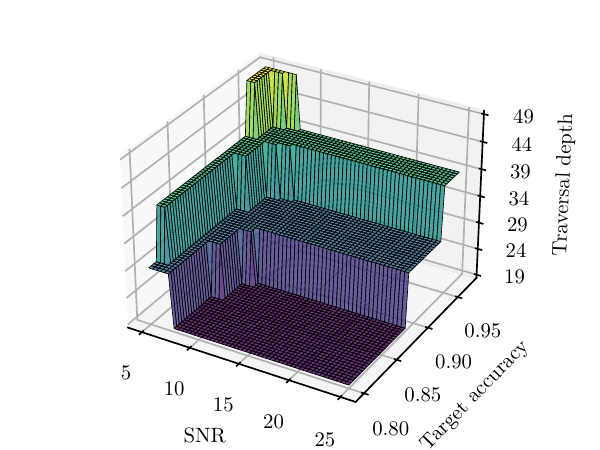}
    \label{vis DQDL}}
    \caption{Performance of practical channel-adaptive AI.}\label{DQDL theory}
    \vspace{-5mm}
  \end{figure}

\subsubsection{Comparison with Learning-Based Optimization}

A representative example of learning-based optimization is~\cite{li2026task}, where an adaptive policy is learned from interaction data instead of being derived from an explicit accuracy model. A similar learning-based idea could be applied to our setting by treating the bit-width \(q\) and traversal depth \(\ell\) as actions. Such a method is flexible when the performance function is difficult to model analytically, but it typically requires interaction data, policy training, and reward or constraint design. In contrast, our method derives an analytical relationship among quantization distortion, traversal depth, and inference accuracy, which leads to an interpretable and low-complexity adaptation rule under explicit air-latency and accuracy constraints. When the environment or AI architecture changes, the model parameters or the accuracy--depth--distortion mapping can be re-estimated using validation data. Conservative accuracy margins or periodic recalibration can also be used to reduce the risk of constraint violations caused by model mismatch.

\section{Experimental Results}\label{experiments}

\subsection{Experimental Settings}\label{Experimental Setting}

\begin{itemize}
    \item \textbf{System and communication setting:} We consider an edge inference system where the backbone model ResNet-152 with $50$ bottleneck layers is split. The first $11$ bottleneck layers are in the edge device, and the remaining $L=39$ bottleneck layers are in the edge server. The \(11/39\) split follows the standard ResNet-152 architecture, whose 50 bottleneck blocks are organized into four residual groups with \(3\), \(8\), \(36\), and \(3\) blocks, respectively~\cite{he2016deep}. Thus, assigning the first \(11=3+8\) blocks to the device and the remaining \(39=36+3\) blocks to the server preserves the residual-group boundary while leaving sufficient server-side depth for adaptive exiting. The set of exit layers that supports adaptive exit is: $9, 19, 24, 29, 34, 37$. The total bandwidth is $B=100\ \mathrm{MHz}$, and the maximal air latency is $T_{\mathrm{max}}=12\ \mathrm{ms}$. The computation speeds of the edge device and the server are set to $0.1$ and $0.5$ TFLOPS, respectively. The maximal bit-width per feature is $Q=32$.
    \item \textbf{Metrics:} EPR and inference accuracy are two metrics to represent the inference performance in the experiments. The former is analogous to the communication rate, illustrating how efficiently the system processes bits in unit inference time. The latter one is obtained by using part of the test dataset to perform the inference process with IID channel realizations and calculating the ratio of correctly classified ones. $2000$ IID inference tasks are performed, and the performance is averaged.
    \item \textbf{Real-world datasets:} CIFAR-10 is used in our experiments. It comprises a training dataset with $50,000$ labeled samples and a test set with $10,000$ labeled samples. The data have $J=10$ classes. The training dataset is used to train the ResNet-152 model and the accompanying CNNs. The test set is randomly split into two sub-datasets with the ratio of $8:2$. The first sub-dataset, which we call the validation dataset, is used to estimate the hyperparameters in the channel-adaptive algorithm. And the remaining part, which we call the (final) test dataset, is used to evaluate the inference performance.
    \item \textbf{Channel-adaptive AI (CA$^2$I):} Consider the algorithm designed in the preceding section, for which the traversal depth for inference adapts to the channel state according to the channel-adaptive AI algorithm. Offline experiments with the validation dataset are carried out to calibrate the accuracy model and obtain the accuracy--depth mapping for each bit-width. Therefore, the bit-width can be calculated based on \eqref{Q_discrete} and the minimal possible exit layer can be obtained based on the mappings by using the binary search. The different sets of exit layers are shown in Table \ref{DQDL resion setting}.
    \item \textbf{Benchmarking scheme:} We consider a non-adaptive AI algorithm as a benchmarking scheme. This is analogous to fixed-rate transmission baselines in adaptive modulation systems~\cite{goldsmith1997variable}, and is used to characterize the gain brought by channel-adaptive bit-width and traversal-depth selection. For this scheme, the bit-width and traversal depth are fixed regardless of the channel states. For example, we can set $q=12$ and $\ell=37$, respectively. Additional combinations of $q$ and $\ell$ are specified later.
\end{itemize}

\begin{table}[t]
\small
\caption{Different sets of exit layers for implementing channel-adaptive AI.}
\begin{adjustbox}{center}
\begin{tabular}{|c|c|}
\hline
\textbf{\# of exit layers} & \textbf{Set of exit layers} \\ \hline
2                & 9, 37                    \\ \hline
3                & 9, 34, 37                \\ \hline
4                & 9, 29, 34, 37            \\ \hline
5                & 9, 24, 29, 34, 37        \\ \hline
6                & 9, 19, 24, 29, 34, 37        \\ \hline
\end{tabular}
\end{adjustbox}
\label{DQDL resion setting}
\end{table}

\subsection{Validation of Inference Accuracy}

\begin{figure*}[t]
\centering
    \subfigure[\(q=8\).]{\includegraphics[width=0.3\linewidth]{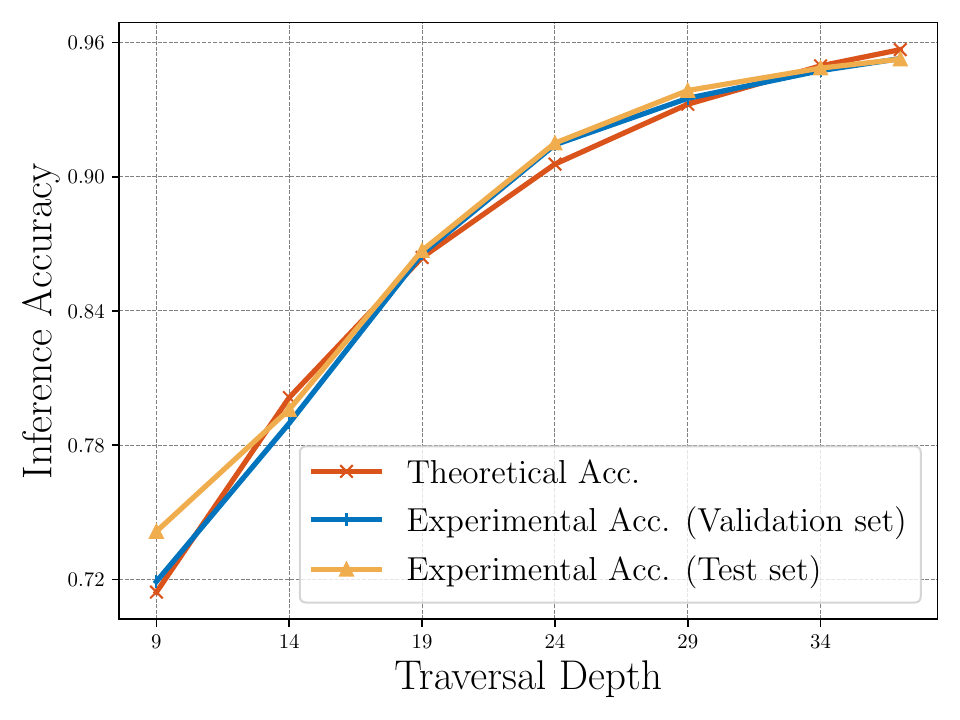}\label{acc_fit_q8}}
    \subfigure[\(q=10\).]{\includegraphics[width=0.3\linewidth]{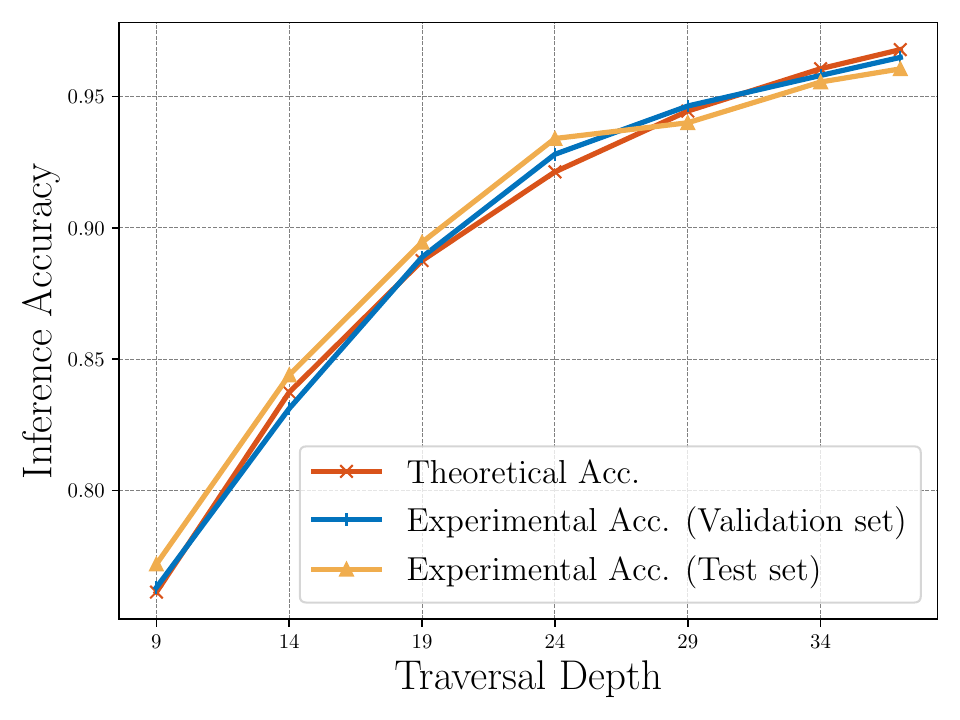}\label{acc_fit_q10}}
    \subfigure[\(q=30\).]{\includegraphics[width=0.3\linewidth]{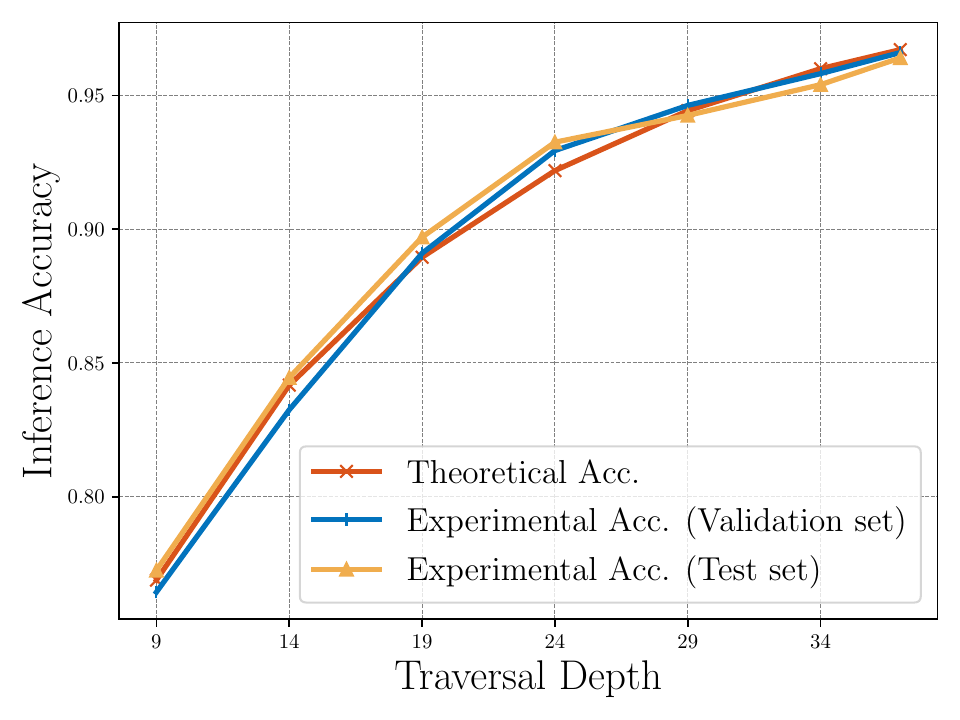}\label{acc_fit_q30}}
    \caption{The theoretical and experimental inference accuracy under different bit-widths.}
    \vspace{5mm}
    \label{acc_theory_empirical}
\end{figure*}

As shown in Fig.~\ref{acc_theory_empirical}, the theoretical accuracy and the experimental accuracy on both the validation set and the test set show a similar increasing trend as the traversal depth increases under different bit-widths. In particular, the case \(q=30\) corresponds to an almost distortion-free setting. The theoretical accuracy is slightly more aligned with the validation-set accuracy because the model is calibrated using the validation set. Nevertheless, since the validation set and the test set are drawn from nearly identical data distributions, their experimental accuracies are also similar. This supports the reliability of the analytical model used in the proposed channel-adaptive algorithm.

Fig.~\ref{acc noise} shows how the inference accuracy is influenced by the bit-width and traversal depth. As shown in Fig.~\ref{acc vs Q}, the inference accuracy increases monotonically w.r.t. the bit-width when the layer is fixed, because a higher bit-width means less distortion on features, which leads to better inference performance. On the other hand, as shown in Fig.~\ref{acc vs l in Q}, when the bit-width is fixed, more layers mean higher inference accuracy, which aligns with the error-free scenarios. This phenomenon is more obvious when the bit-width is higher, because the distortion caused by extremely low bit-width yields high randomness.

\begin{figure}[t]
\centering
    \subfigure[Inference accuracy w.r.t. $q$.]{\includegraphics[width=0.48\linewidth]{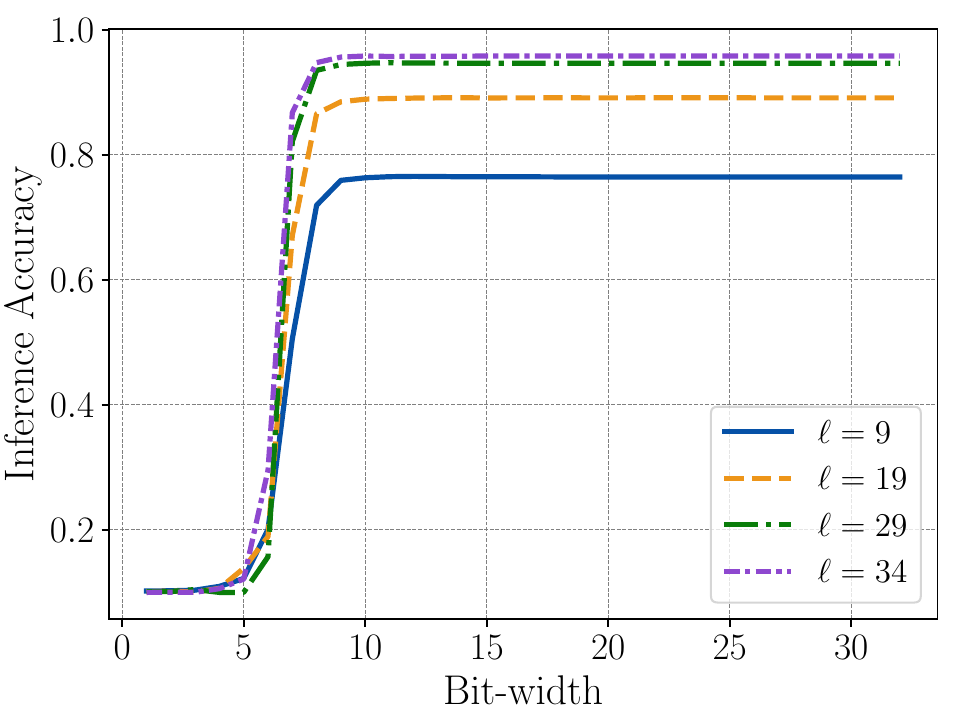}\label{acc vs Q}}
    \subfigure[Inference accuracy w.r.t. $\ell$.]{\includegraphics[width=0.48\linewidth]{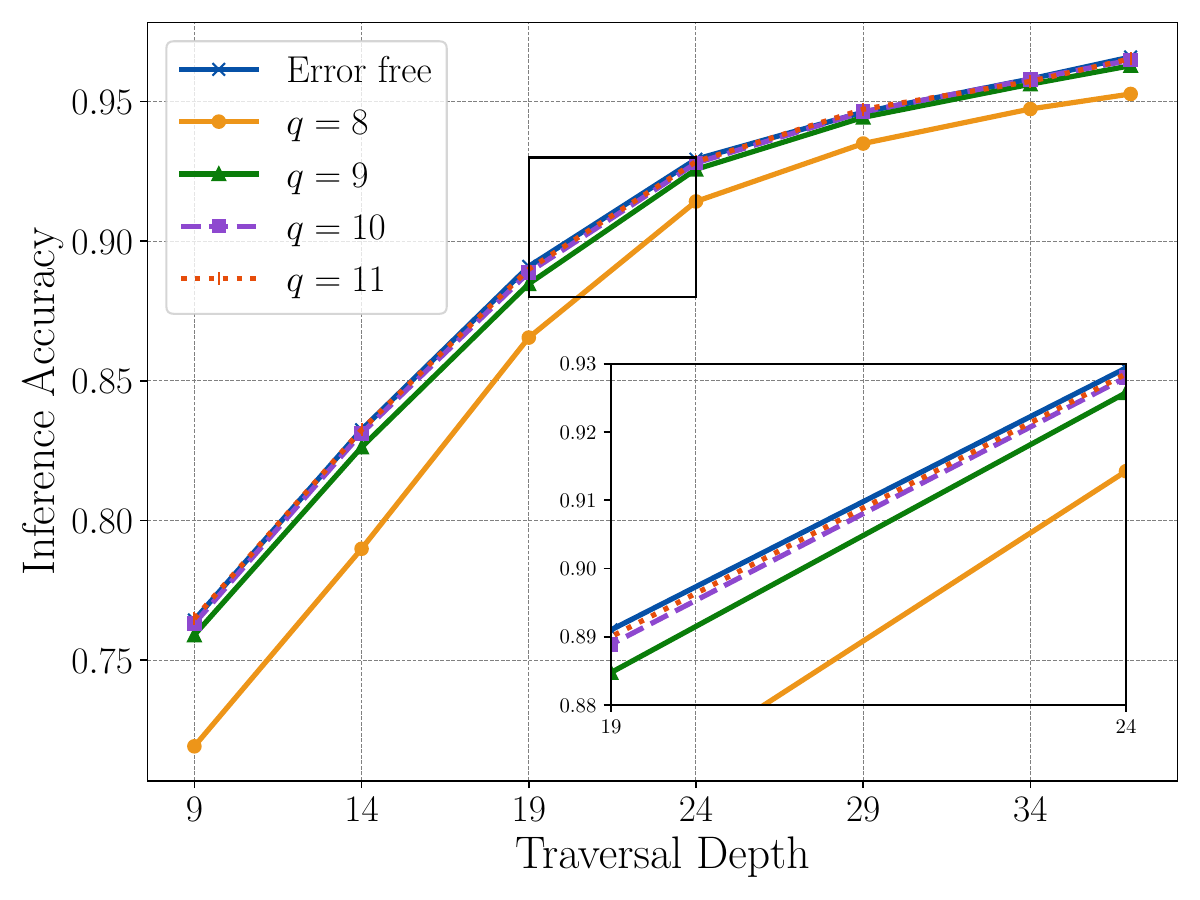}\label{acc vs l in Q}}
    \caption{The performance of the inference model with different bit-width and traversal depth.}
    \vspace{-5mm}
    \label{acc noise}
  \end{figure}

\subsection{Inference Performance of Channel-Adaptive AI}

We compare the performance of channel-adaptive AI with its counterpart, the non-adaptive AI algorithm, as shown in Fig.~\ref{CA2I performance}. In Fig.~\ref{EPE DQDL}, it can be seen that the EPR of channel-adaptive AI increases monotonically w.r.t. both the SNR and the number of exit layers. This is because higher SNRs allow for higher bit-width and smaller traversal depth, which collaboratively increase the EPR. And when the number of exit layers increases, there is potential that fewer layers could be used while satisfying the accuracy threshold, which increases the EPR. As for the non-adaptive AI algorithm, there is a possibility that the fixed bit-width cannot guarantee the communication constraint, which leads to zero EPR when the SNR is low. Omitting the fixed bit-width, the non-adaptive AI algorithm can be considered as a special case of channel-adaptive AI, i.e., the number of exit layers is 1. Therefore, its performance is not as good as the channel-adaptive AI with more exit layers. Specifically, when the SNR is $15 
\ \mathrm{dB}$, channel-adaptive AI with 5 exit layers demonstrates a $34.3\%$ higher EPR compared with the non-adaptive algorithm. Furthermore, at an SNR of $25 
\ \mathrm{dB}$, the performance gain is over $100\%$.
Fig.~\ref{accuracy DQDL} shows the actual inference accuracy of the schemes in different SNRs. In low-SNR scenarios, the inference accuracy has not reached the target accuracy, as the distortion caused by the channel cannot be mitigated even when all the layers are utilized. Another observation is that in high-SNR scenarios, increasing the number of exit layers decreases the inference accuracy, while still satisfying the accuracy requirement. This illustrates the adaptability of channel-adaptive AI, which sacrifices some accuracy when necessary to achieve a higher EPR. As for the non-adaptive AI algorithm, since the features cannot be transmitted successfully within the transmission period, random guessing is used, which yields an inference accuracy of $1/J=0.1$.

\begin{figure}[t]
\centering
        \subfigure[The EPR.]{\includegraphics[width=0.48\linewidth]{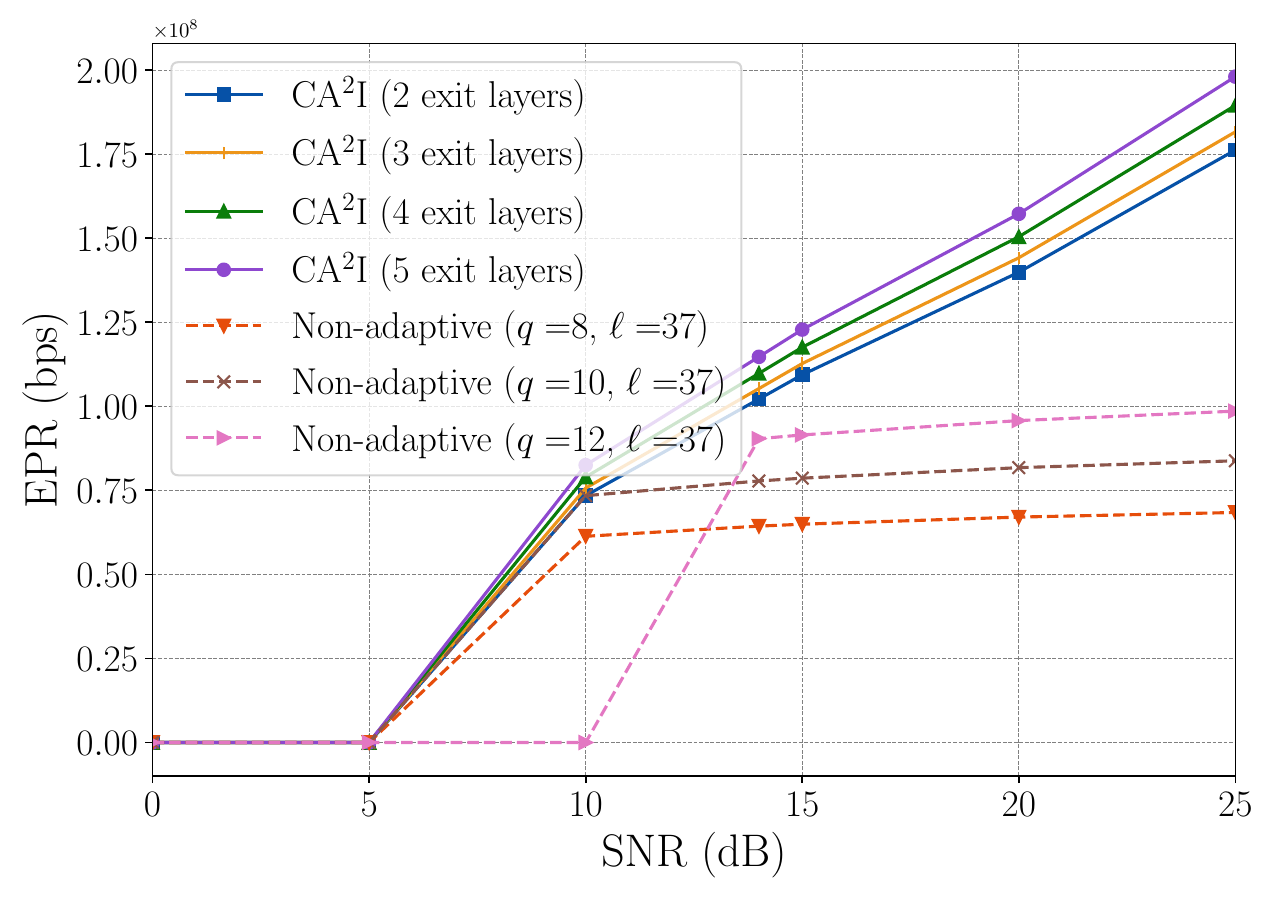}\label{EPE DQDL}}
    \subfigure[The inference accuracy.]{\includegraphics[width=0.48\linewidth]{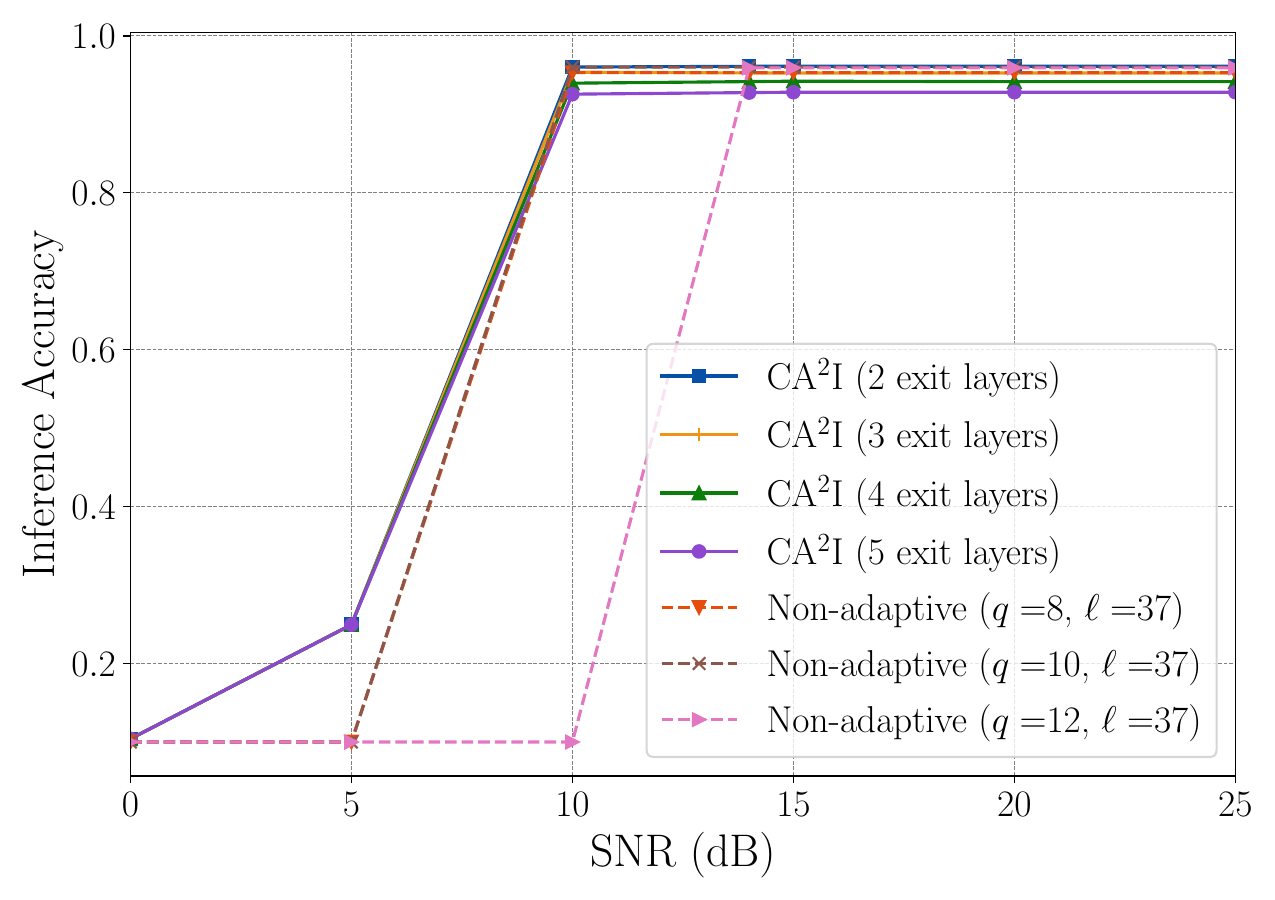}\label{accuracy DQDL}}
    \caption{The performance of channel-adaptive AI algorithm with different numbers of exit layers. The target accuracy is $90\%$.}
    \vspace{-5mm}
    \label{CA2I performance}
  \end{figure}

Fig.~\ref{DQDL EPE different acc0} illustrates the performance of channel-adaptive AI across different target accuracy levels. The curves align with the expectation that a lower target accuracy permits a higher EPR. Specifically, when the SNR is $25\ \mathrm{dB}$, reducing the target accuracy from $0.95$ to $0.9$ and $0.85$ results in associated EPR increases of $9.04\%$ and $14.2\%$, respectively. 

\begin{figure}[!]
\centering
{\includegraphics[width=0.6\linewidth]{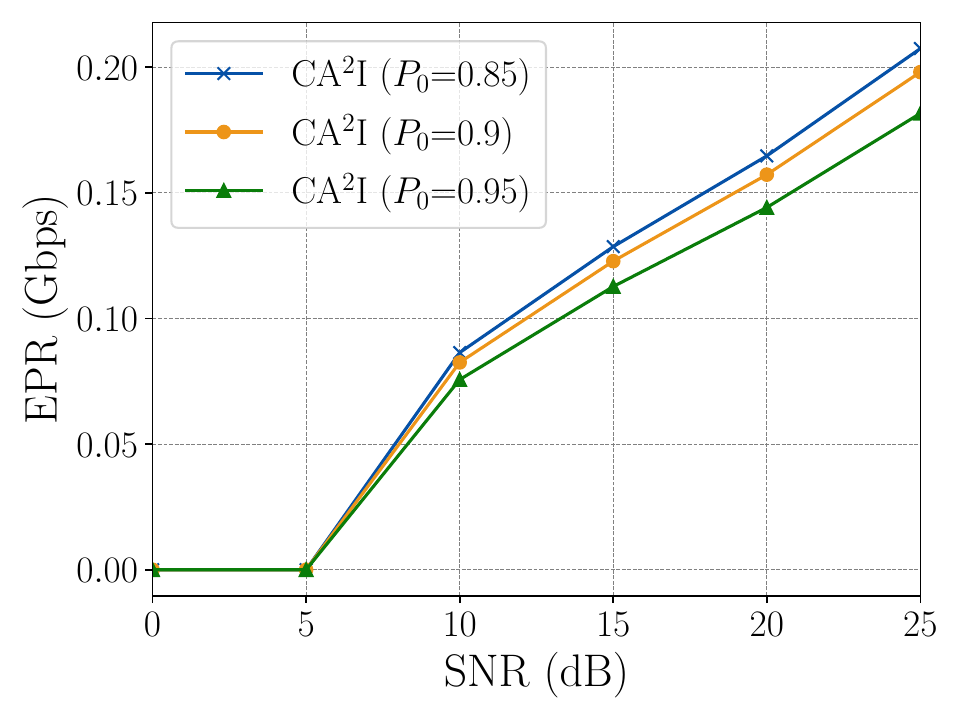}}
\caption{The performance of channel-adaptive AI algorithm with different target accuracy. The number of exit layers is set to $6$.}
\vspace{-5mm}
\captionsetup{justification=justified}
\label{DQDL EPE different acc0}
  \end{figure}

\section{Concluding Remarks}\label{concluding remarks}

In this paper, we present a novel channel-adaptive AI framework that adaptively controls the quantization and computational complexity according to the channel state. Such a computation-communication integrated design maximizes the inference throughput while providing a guarantee on air latency and inference accuracy, thereby supporting latency-sensitive edge AI applications. To develop the framework, we propose a tractable model for inference accuracy, wherein the intermediate angular features are modeled as a mixture of von Mises distributions. Experiments demonstrate the superiority of the channel-adaptive AI algorithm compared with its conventional non-adaptive counterpart.

Several directions for further research merit exploration. First, the effects of fast fading can be incorporated into the communication model. In this context, metrics such as EPR and theoretical accuracy can be replaced with expected values relative to channel realization; the optimization problem can be reformulated accordingly. Additionally, techniques to mitigate unsuccessful inference due to channel outage need to be designed. Furthermore, the joint design of adaptive transmission (i.e., modulation, power control, and beamforming) and adaptive computational complexity is also a promising direction for future work. Additionally, although the proposed MvM-based performance model is tailored to classification tasks, the high-level principle of channel-adaptive AI can be extended to other inference tasks with intermediate exits. For tasks such as object detection, semantic segmentation, and regression, deeper traversal generally provides higher task performance at the cost of increased computation. Developing task-specific performance models for them and extending the proposed channel-adaptive algorithm to these tasks are left for future work.

\section*{Appendix}
\appendices
\renewcommand{\thesubsection}{\Alph{subsection}}

\subsection{Parameter Estimation of MvM}\label{Parameter Estimation of MvM}

Getting the exact distribution of the angular features necessitates the estimation of $\kappa_{\Delta,\ell}$, with finite samples.
We start by investigating the $N_j$ samples for class $j$: $\tilde\theta_{\ell,1}^{(j)}, \tilde\theta_{\ell,2}^{(j)}, \dots, \tilde\theta_{\ell,N_j}^{(j)}$.
Referring to \cite{kutil2012biased}, a biased estimation of the $\kappa^{(j)}_{\Delta,\ell}$ is given by
\begin{equation}
\hat{\kappa}^{(j)}_{\Delta,\ell}=A^{-1}\biggl(\frac{1}{N_j}\sqrt{\big(\sum_{n=1}^{N_j} \cos(\tilde\theta_{\ell,n}^{(j)})\big)^2+\big(\sum_{n=1}^{N_j}\sin(\tilde\theta_{\ell,n}^{(j)})\big)^2}\biggl).\label{estimate kappa}
\end{equation}

To align with the identical $\kappa$ assumption, we further estimate the $\kappa_{\Delta,\ell}$ in the MvM model as
\begin{equation}
\hat{\kappa}_{\Delta,\ell} = \frac{1}{J}\sum_{j=1}^J\hat{\kappa}^{(j)}_{\Delta,\ell}.\label{estimate kappa l}
\end{equation}

\subsection{Proof of Proposition \ref{prop distorted distribution}}
\label{Proof of P1}

In this subsection, we prove that the channel distorted angular feature $\tilde{\theta}_\ell|j$ follows the vM distribution for arbitrary quantization level and traversal depth. Let $f_\ell(\cdot) =(F_{\mathrm{cnn},\ell} \circ F_\ell \circ \cdots \circ F_1)(\cdot)$ denote the composite mapping from feature $\tilde{\mathbf{z}}$ to angular feature $\tilde{\theta}_\ell$. The details of the proof are narrated as follows.

\subsubsection{Differentiability Assumption}

We assume the composite function $f_\ell(\cdot)$ is differentiable w.r.t. the input $\mathbf{z}$. Given that $f_\ell(\cdot)$ comprises the backbone network $F_\ell \circ \cdots \circ F_1$ and the intermediate CNN $F_{\mathrm{cnn},\ell}$, we examine the differentiability of each component as follows.
\begin{itemize}
    \item \textbf{Backbone:} The ResNet-152 backbone employed in this work comprises standard convolutional and pooling layers. Although ReLU activations and max-pooling operations are non-differentiable at specific singularities, these occur on sets of measure zero and are resolved via subgradients. Consequently, we treat the backbone as differentiable.
    \item \textbf{Intermediate CNN:} 
    For $F_{\mathrm{cnn},\ell}(\cdot)=(\mathrm{atan}2\circ \phi_\ell)(\cdot)$, the projection $\phi_\ell$ inherits the differentiability of the backbone network as it shares the same operations. However, the $\mathrm{atan}2$ function operator introduces singularities: it is undefined at the origin and exhibits a discontinuity along the negative real axis (the branch cut), where the output jumps between $\pi$ and $-\pi$. The singularity at the origin is negligible as it constitutes a set of measure zero. The discontinuity can be solved by introducing an auxiliary function $f^\dagger_\ell(\cdot)$ that serves as a locally continuous surrogate for $f_\ell(\cdot)$. Specifically, to compensate for the jump at the branch cut, we define $f^\dagger_\ell(\cdot) \triangleq f_\ell(\cdot) - 2\pi$ for transitions from $-\pi$ to $\pi$, and $f^\dagger_\ell(\cdot) \triangleq f_\ell(\cdot) + 2\pi$ for the reverse transition. Accordingly, we define the gradient at the branch cut as $\nabla f_\ell(\mathbf{z}) \triangleq \nabla f^\dagger_\ell(\mathbf{z})$. Overall, the gradient is given by
\begin{equation}
\begin{split}
    \mathbf{a}_\ell= & \ \nabla f_\ell(\mathbf{z}) \\
    =&\left( \prod_{k=1}^{\ell} \left( \frac{\partial F_k}{\partial \mathbf{z}_{k-1}} \right)^\top \right) \left( \frac{\partial \phi_\ell}{\partial \mathbf{z}_\ell} \right)^\top \left( \frac{1}{\|\mathbf{v}_\ell\|^2} \begin{bmatrix} -v_y \\ v_x \end{bmatrix}\right), 
\end{split}
     \label{chain rule}
\end{equation}
where $\mathbf{z}_{0}\triangleq \mathbf{z}$, $\mathbf{z}_{k}=F_k(\mathbf{z}_{k-1})$ for $k=1,\cdots,\ell$, and $\mathbf{v}_\ell = \phi_\ell(\mathbf{z}_\ell)=[v_x,v_y]^\top$ denotes the 2D projected feature. The last term in \eqref{chain rule} comes from the standard derivative of $\mathrm{atan2}$~\cite{ahlfors1979complex}.
\end{itemize}
Based on the preceding analysis, we proceed under the assumption that $f_\ell(\cdot)$ is differentiable almost everywhere.

\subsubsection{Proof of the Distorted Angular Feature Distribution}

Let $\boldsymbol{\epsilon} = \tilde{\mathbf{z}} - \mathbf{z}$ denote the quantization distortion vector, modeled as IID uniform random variables with zero mean and variance $\sigma_\Delta^2$. Let $\bar{\theta}_\ell \triangleq f_\ell(\mathbf{z})$ represent the undistorted angular feature (when $\sigma_\Delta^2=0$), while $\tilde{\theta}_\ell=f_\ell(\tilde{\mathbf{z}})$ is the distorted one.
We analyze the distribution of $\tilde{\theta}_\ell$ using a first-order Taylor expansion in the neighborhood of $\mathbf{z}$:
\begin{equation}
   \begin{split}
\tilde{\theta}_\ell = f_\ell(\tilde{\mathbf{z}}) &= f_\ell(\mathbf{z}+\boldsymbol{\epsilon}) \\
       &\overset{(a)}\approx \Big( \bar\theta_\ell + \mathbf{a}_\ell^T\boldsymbol{\epsilon} \Big) \bmod 2\pi \\
       &\overset{(b)}= \bigg( \bar\theta_\ell + \underbrace{\big(\mathbf{a}_\ell^T\boldsymbol{\epsilon} \bmod 2\pi\big)}_{\triangleq \epsilon_\ell} \bigg) \bmod 2\pi,
   \end{split}
   \label{first-order Taylor}
\end{equation}
where (a) comes from the first-order Taylor approximation in the angular domain. Note that when $f_\ell(\mathbf{z})$ approaches $\pm\pi$, the term $\bar\theta_\ell + \mathbf{a}_\ell^T\boldsymbol{\epsilon}$ may fall outside the range $(-\pi, \pi]$. The modulo operation $\bmod$ corrects this by wrapping the value back into the valid domain, ensuring angular periodicity. For values already within bounds, the operation leaves them unchanged. (b) comes from the modular arithmetic identity $(a + b) \bmod c = ((a \bmod c) + (b \bmod c)) \bmod c$~\cite{hua2012introduction}.
Assuming the dimensionality of $\boldsymbol{\epsilon}$ is sufficiently large and that Lindeberg's condition holds, the Central Limit Theorem (CLT) suggests that $\mathbf{a}_\ell^T\boldsymbol{\epsilon}$ converges to a Gaussian distribution $\mathcal{N}\big(0, \sigma_\Delta^2 a_\ell\big)$~\cite{durrett2019probability}. Consequently, the term $\epsilon_\ell$ follows a wrapped Gaussian distribution. To facilitate the approximation of this distribution, we introduce Lemma~\ref{wrapped Gaussian}.

\begin{lemma}[Approximation of wrapped Gaussian with vM distribution~\cite{mardia2009directional}]\label{wrapped Gaussian}
    Consider a random variable $\beta \sim \mathcal{N}(\mu, \sigma^2)$. Its wrapped version, defined as $\omega \triangleq \beta \pmod{2\pi}$, can be approximated by a vM distribution:
\begin{equation}
    \psi \sim v\mathcal{M}\Big(\mu \bmod 2\pi, \, A^{-1}\big(\exp(-\sigma^2/2)\big)\Big).
\end{equation}
\end{lemma}

\begin{figure}[ht]
\centering
{\includegraphics[width=1\linewidth]{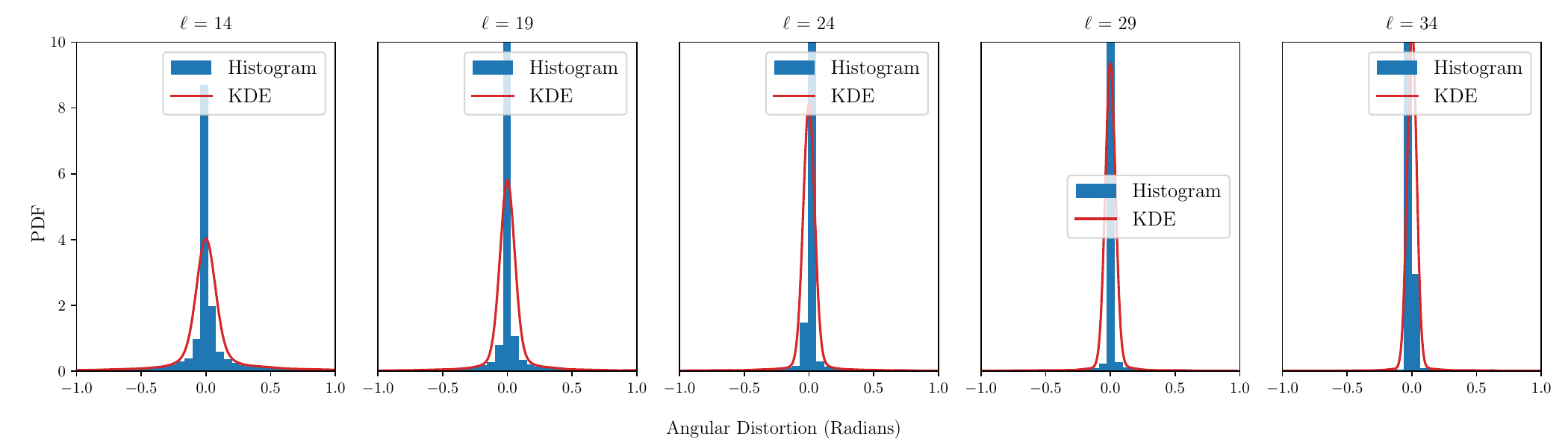}}
\caption{The distribution of distortion in angular space w.r.t. different traversal depths.}
\captionsetup{justification=justified}
\label{Distrib of distortion}
  \end{figure}

Based on Lemma~\ref{wrapped Gaussian}, we can approximate $\epsilon_\ell$ with a vM distribution $v\mathcal{M}(0,\rho_{\Delta,\ell})$,
where $\rho_{\Delta,\ell}=A^{-1}\Big(\exp\big(-\sigma_\Delta^2 a_\ell/2\big)\Big)$. We validate this approximation in Fig.~\ref{Distrib of distortion} by plotting the distribution of $\epsilon_\ell$. According to \eqref{first-order Taylor}, where $\bar\theta_\ell \sim v\mathcal{M}(\mu_j, \bar\kappa_\ell)$ and $\epsilon_\ell \sim v\mathcal{M}(0, \rho_{\Delta,\ell})$, we apply the convolution approximation from \cite{mardia2009directional}, which yields $\tilde{\theta}_\ell \sim v\mathcal{M}\Big(\mu_j, A^{-1}\big(A(\bar\kappa_\ell)A(\rho_{\Delta,\ell})\big)\Big)$, thereby completing the proof.

\subsubsection{Approximation of Expected Gradient Norm $a_\ell$}

Note that the parameter $a_\ell$ depends on the complex internal weights of $f_\ell(\cdot)$, making it intractable to model explicitly. According to numerous experimental results, we observe that $a_\ell$ can be approximated by an exponential function as \eqref{al} (see Fig.~\ref{al vs l} for validation). 

\begin{figure}[ht]
\centering
{\includegraphics[width=0.6\linewidth]{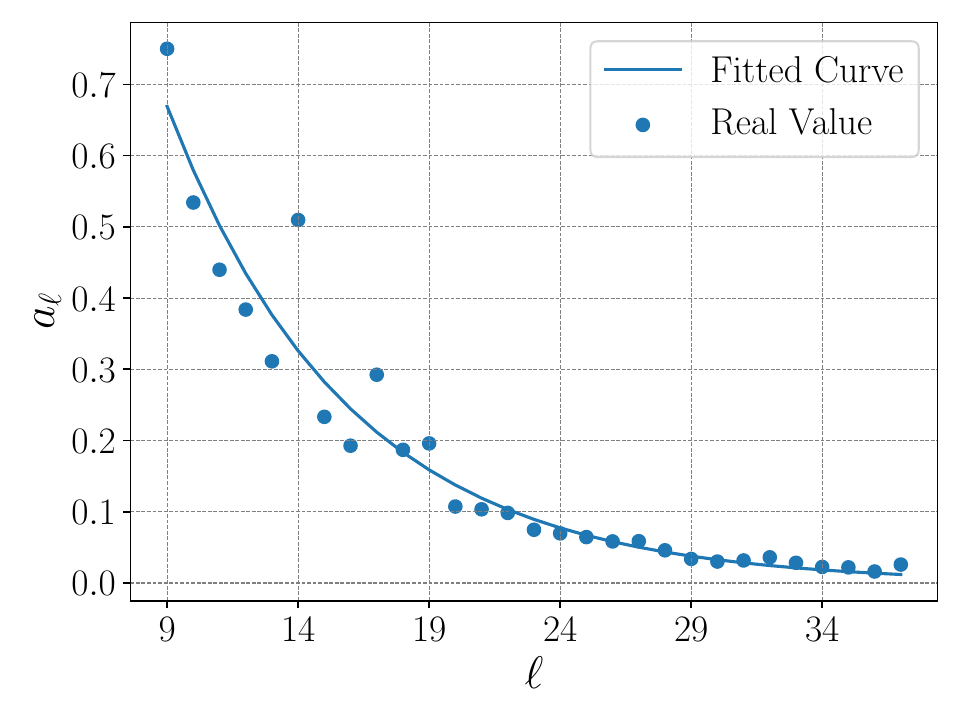}}
\caption{$a_\ell$ and the fitted curve.}
\vspace{-5mm}
\captionsetup{justification=justified}
\label{al vs l}
  \end{figure}

\subsection{Proof of Corollary \ref{monotonicity accuracy}}
\label{proof monotonicity accuracy}

We prove the two claims in Corollary \ref{monotonicity accuracy} as follows.

\subsubsection{Monotonicity of \eqref{accuracy} w.r.t. $J$ and $\kappa_{\Delta,\ell}$}

The first part is trivial because the integrand is positive, and the increasing $J$ shrinks the integration interval, which naturally reduces the integral.
To prove the second part, we begin by analyzing the integrand defined as $f(x,\kappa) \triangleq \frac{\exp(\kappa \cos x)}{\pi I_0(\kappa)}$. For $\forall \ 0<\kappa_1<\kappa_2$, we define 
$g(x)\stackrel{\triangle}= \ \ln f(x,\kappa_2)-\ln f(x,\kappa_1)=\ (\kappa_2-\kappa_1)\cos(x)-\big(\ln I_0(\kappa_2)-\ln I_0(\kappa_1)\big).$
Observe that $\big(\ln I_0(x)\big)'=\frac{I_1(x)}{I_0(x)}=A(x)$, then we have $\ln I_0(\kappa_2)-\ln I_0(\kappa_1)=A(\xi)(\kappa_2-\kappa_1)$ for some $\xi\in (\kappa_1, \kappa_2)$ by Lagrange's Mean Value Theorem. Therefore
$g(0)=\big(1-A(\xi)\big)(\kappa_2-\kappa_1)>0$ since $A(x)<1$ for $\forall x>0$. And $g(\pi)=-(\kappa_2-\kappa_1)-\big(\ln I_0(\kappa_2)-\ln I_0(\kappa_1)\big)<0$ because $\ln I_0(x)$ is a monotonically increasing function. Additionally, $g(x)$ decreases monotonically in $(0,\pi)$. Hence, there exists a unique $x_0\in(0,\pi)$ such that $g(x_0)=0$. Therefore, when $0<x<x_0$, $g(x)>0$  and thus $f(x,\kappa_2)>f(x,\kappa_1)$, and when $x_0<x<\pi$, $g(x)<0$ and thus $f(x,\kappa_2)<f(x,\kappa_1)$.

If $x_0\geq\frac{\pi}{J}$, we have $\int_{0}^{\frac{\pi}{J}}f(x,\kappa_2)\, dx-\int_{0}^{\frac{\pi}{J}}f(x,\kappa_1)\, dx=\int_{0}^{\frac{\pi}{J}}\big(f(x,\kappa_2)-f(x,\kappa_1)\big)\, dx>0$.
Otherwise, we have 
\begin{equation}
    \begin{split}
        &\int_{0}^{\frac{\pi}{J}}f(x,\kappa_2)\, dx-\int_{0}^{\frac{\pi}{J}}f(x,\kappa_1)\, dx\\
        =&\big(1-\int_{\frac{\pi}{J}}^{\pi}f(x,\kappa_2)\, dx\big)-\big(1-\int_{\frac{\pi}{J}}^{\pi}f(x,\kappa_1)\, dx\big)\\
        =&\int_{\frac{\pi}{J}}^{\pi}\big(f(x,\kappa_1)-f(x,\kappa_2)\big)\, dx>0.
    \end{split}
\end{equation}
Since for both situations we have $\int_{0}^{\frac{\pi}{J}}f(x,\kappa_2)\, dx-\int_{0}^{\frac{\pi}{J}}f(x,\kappa_1)\, dx>0$, we conclude that $\int_{0}^{\frac{\pi}{J}}f(x,\kappa_2)\, dx>\int_{0}^{\frac{\pi}{J}}f(x,\kappa_1)\, dx$ for $\forall \ 0<\kappa_1<\kappa_2$. This completes the proof. 

\subsubsection{Monotonicity of \eqref{accuracy} w.r.t. $\sigma_\Delta^2$ and $\ell$}

Both $A(c_1 \ell + c_2)$ and $\exp\big(-\sigma_\Delta^2 c_3\exp(-c_4\ell)/2\big)$ are positive and increase monotonically w.r.t. $\ell$. Hence, their product, denoted as $D(\sigma_\Delta^2,\ell)$, increases monotonically w.r.t. $\ell$. Additionally, $D(\sigma_\Delta^2,\ell)$ decreases monotonically w.r.t. $\sigma_\Delta^2$. Therefore, $\kappa_{\Delta,\ell}=A^{-1}\big(D(\sigma_\Delta^2,\ell)\big)$ decreases monotonically w.r.t. $\sigma_\Delta^2$ and increases monotonically w.r.t. $\ell$. Combining the result about the monotonicity of \eqref{accuracy} w.r.t. $\kappa_{\Delta,\ell}$, the proof is completed.

\subsection{Proof of Corollary \ref{range of accuracy}}
\label{proof of range of accuracy}

The first part is straightforward: as $\sigma_\Delta^2 \to \infty$, we have $\kappa \to 0$. Given the continuity of $P(\sigma_\Delta^2,\ell)$ with respect to $\kappa_{\Delta,\ell}$, substituting $\kappa_{\Delta,\ell}=0$ into \eqref{accuracy} completes the proof.
For the second part, we observe that as $\ell \to \infty$, $\kappa_{\Delta,\ell} \to \infty$. We proceed by establishing the following conclusion:
\begin{equation}
I_0(\kappa;a)\approx
\frac{e^\kappa}{\sqrt{2\pi \kappa}}
\mathrm{erf}\bigg(a\sqrt{\frac{\kappa}{2}}\bigg),\label{appriximation of I k a}
\end{equation}
where $I_0(\kappa;a)\stackrel{\triangle}=\frac{1}{\pi}\int_{0}^{{a}}\exp{(\kappa \cos{(x)})}\,dx$ and $\mathrm{erf}(x)\stackrel{\triangle}=\frac{2}{\sqrt{\pi}}\int_{0}^{{x}}\exp{(-\theta^2)}\,d\theta$.
Considering Taylor's series for approximation, we have $\cos(\theta)\approx1-\theta^2/2$, and therefore
\begin{equation}
\begin{split}
I_0(\kappa;a)&\approx\frac{e^{\kappa}}{\pi} \int_{0}^{{a}}\exp{(-\kappa\theta^2/2)}\,d\theta\\ &\overset{(a)}=\frac{e^{\kappa}}{\pi}\sqrt{\frac{2}{\kappa}}\int_{0}^{{a\sqrt{\frac{\kappa}{2}}}}\exp{(-\mu^2)}\,d\mu\\
    &=\frac{e^\kappa}{\sqrt{2\pi \kappa}}
\mathrm{erf}\bigg(a\sqrt{\frac{\kappa}{2}}\bigg),
\end{split}
\end{equation}
where (a) holds by replacing variable $\theta$ with $\mu\sqrt{\frac{2}{\kappa}}$. Thereby, \eqref{appriximation of I k a} is proved.
Hence,
\begin{equation}
    \begin{split}
\lim_{\kappa_{\Delta,\ell}\to\infty}P(\sigma_\Delta^2,\ell)
&= \lim_{\kappa_{\Delta,\ell}\to\infty}\frac{I_0(\kappa_{\Delta,\ell};\frac{\pi}{J})}{I_0(\kappa_{\Delta,\ell};\pi)}\\
&=\lim_{\kappa_{\Delta,\ell}\to\infty}\frac{\mathrm{erf}\bigg(\frac{\pi}{J}\sqrt{\frac{\kappa_{\Delta,\ell}}{2}}\bigg)}{\mathrm{erf}\bigg(\pi\sqrt{\frac{\kappa_{\Delta,\ell}}{2}}\bigg)}=1,
    \end{split}
\end{equation}
which completes the proof.

\subsection{Proof of Proposition \ref{asymptotic behavior}}\label{proof asymptotic behavior}

We first consider the asymptotic analysis of $1-P(\sigma_\Delta^2,\ell)$ w.r.t. $\kappa_{\Delta,\ell}$. 
\begin{equation}
    \begin{split}
        1-\frac{I_0(\kappa_{\Delta,\ell};\frac{\pi}{J})}{I_0(\kappa_{\Delta,\ell};\pi)}&\sim 1-\frac{\mathrm{erf}\bigg(\frac{\pi}{J}\sqrt{\frac{\kappa_{\Delta,\ell}}{2}}\bigg)}{\mathrm{erf}\bigg(\pi\sqrt{\frac{\kappa_{\Delta,\ell}}{2}}\bigg)}\\
        &\overset{(a)}{\sim}1-\mathrm{erf}\bigg(\frac{\pi}{J}\sqrt{\frac{\kappa_{\Delta,\ell}}{2}}\bigg)\\
        &\overset{(b)}{\sim}\frac{\sqrt{2}\, J}{\pi^{3/2} \sqrt{\kappa_{\Delta,\ell}}}\,
\exp\left( -\frac{\pi^2}{2 J^2} \kappa_{\Delta,\ell} \right),\label{analysis of 1-P}
    \end{split}
\end{equation}
where (a) holds because the $\mathrm{erf}\big(\pi\sqrt{\frac{\kappa}{2}}\big)$ saturates quickly and approaches $1$ when $\kappa$ is large. (b) holds because of the asymptotic expansion of $\mathrm{erf}(\cdot)$ as $\mathrm{erf}(z) \sim 1 - \frac{e^{-z^2}}{z\sqrt{\pi}}$\cite{abramowitz1964handbook}.
Note that when $\sigma^2_\Delta=0$, we have $\kappa_{\Delta,\ell}=c_1 \ell + c_2$. Substituting it into \eqref{analysis of 1-P} yields the stated results.

\bibliographystyle{ieeetr}
\bibliography{Ref}

\end{document}